\documentclass[smallextended]{svjour3}
\usepackage{graphicx}
\usepackage{mathptmx} 
\usepackage{latexsym}
\usepackage{amsmath}
\usepackage{amssymb}
\usepackage{mathtools}
\usepackage[table]{xcolor}
\usepackage{graphicx}
\usepackage{verbatim}
\usepackage{xspace}
\usepackage{cite}
\usepackage{enumerate}
\usepackage{nicefrac}
\usepackage{tikz}
\usetikzlibrary{positioning,chains,fit,shapes,calc,arrows}

\usepackage{placeins}

\let\doendproof\endproof
\renewcommand\endproof{~\hfill\qed\doendproof}

\newcommand{\ab}[1] {\left\vert #1\right\vert}
\DeclarePairedDelimiter{\ceil}{\lceil}{\rceil}
\DeclarePairedDelimiter{\floor}{\lfloor}{\rfloor}

\newcommand{\no}[1] {\left\vert #1\right\vert_1}
\newcommand{\nz}[1] {\left\vert #1\right\vert_0}
\newcommand{\advice}{\ensuremath{b}\xspace}
\newcommand{\iphi}{\ensuremath{I_{\varphi}}\xspace}
\newcommand{\mincost}{\ensuremath{L_1}\xspace}
\newcommand{\minbits}{\ensuremath{d}\xspace}
\newcommand{\mzero}{\ensuremath{m_0}\xspace}
\newcommand{\mone}{\ensuremath{m_1}\xspace}
\newcommand{\dzero}{\ensuremath{d_0}\xspace}
\newcommand{\done}{\ensuremath{d_1}\xspace}
\newcommand{\dstart}{\ensuremath{d}\xspace}

\newcommand{\Bincr}{\ensuremath{\mathit(B\ref{binomincr})}\xspace}
\newcommand{\Bminus}{\ensuremath{\mathit(B\ref{binomminus})}\xspace}
\newcommand{\Bfrac}{\ensuremath{\mathit(B\ref{binomfrac})}\xspace}

\DeclareMathOperator{\smallmod}{mod}

\def\OPT{\ensuremath{\textsc{Opt}}\xspace}
\def\ALG{\ensuremath{\textsc{Alg}}\xspace}
\def\ALGp{\ensuremath{\textsc{Alg}'}\xspace}
\def\ORACLE{\ensuremath{\texttt{O}}\xspace}
\def\ORACLEp{\ensuremath{\texttt{O}'}\xspace}

\def\Wmin{\ensuremath{\mathsf{AOC}}\xspace}
\def\Wmax{\ensuremath{\mathsf{AOC}}\xspace}

\def\Wminc{\ensuremath{\mathsf{AOC}\text{-complete}}\xspace}
\def\Wmaxc{\ensuremath{\mathsf{AOC}\text{-complete}}\xspace}

\def\W{\ensuremath{\mathsf{AOC}}\xspace}
\def\Wc{\ensuremath{\mathsf{AOC}\text{-complete}}\xspace}
\def\Wcomplete{\ensuremath{\mathsf{AOC}\text{-complete}}\xspace}

\def\sq{\sqsubseteq}

\def\oldS{{\sc SG}\xspace}
\def\minS{{\sc min\-ASG}\xspace}
\def\minSu{{\sc min\-ASGu}\xspace}
\def\minSk{{\sc min\-ASGk}\xspace}
\def\maxS{{\sc max\-ASG}\xspace}
\def\maxSu{{\sc max\-ASGu}\xspace}
\def\maxSk{{\sc maxASGk}\xspace}
\def\S{{\sc ASG}\xspace}
\def\Su{{\sc ASGu}\xspace}
\def\Sk{{\sc ASGk}\xspace}

\def\P{{\textsc P}\xspace}
\def\VC{{\sc Online Vertex Cover}\xspace}
\def\IS{{\sc Online Independent Set}\xspace}
\def\DS{{\sc Online Dominating Set}\xspace}
\def\DPA{{\sc Online Disjoint Path Allocation}\xspace}
\def\CF{{\sc Online Cycle Finding}\xspace}
\def\SC{{\sc Online Set Cover}\xspace}
\def\KS{{\sc Online Uniform Knapsack}\xspace}
\def\OM{{\sc Online Matching}\xspace}

\begin{document}
\def\asydir{}

\title{The Advice Complexity of a Class of Hard Online Problems\thanks{A preliminary
version of this paper appeared in the proceedings of the 32nd International 
Symposium on Theoretical Aspects of Computer Science (STACS 2015), 
\emph{Leibniz International Proceedings in Informatics} 30: 116-129, 2015.}
\thanks{This work was partially supported by the Villum Foundation and the Danish Council for Independent Research,
Natural Sciences.}}
\author{Joan Boyar \and Lene M. Favrholdt \and Christian Kudahl \and Jesper W. Mikkelsen}
\date{\today}
\institute{J. Boyar \and L.M. Favrholdt \and C. Kudahl \and J.W. Mikkelsen
\at Department of Mathematics and Computer Science,
University of Southern Denmark, 5230 Odense M, Denmark\\Tel.: +45 6550-2338,
\\\email{\{joan,lenem,jesperwm,kudahl\}@imada.sdu.dk}}

\maketitle

\begin{abstract}
The advice complexity of an online problem is a measure of how much knowledge of the future an online algorithm needs in order to achieve a certain competitive ratio. 
Using advice complexity, we define the first online complexity class, AOC.
The class includes independent set, vertex cover, dominating set, and several others as complete problems.
AOC-complete problems are hard, since a single wrong answer by the online algorithm can have devastating consequences. For each of these problems, we show that $\log\left(1+(c-1)^{c-1}/c^{c}\right)n=\Theta (n/c)$ bits of advice are necessary and sufficient (up to an additive term of $O(\log n)$) to achieve a competitive ratio of $c$. 

The results are obtained by introducing a new string guessing problem related to those of Emek et al.~(TCS 2011) and B\"ockenhauer et al.~(TCS 2014). It turns out that this gives a powerful but easy-to-use method for providing both upper and lower bounds on the advice complexity of an entire class of online problems,
the AOC-complete problems.

Previous results of Halld\'orsson et al.~(TCS 2002) on online independent set, in a related model, imply that the advice complexity of the problem is $\Theta (n/c)$. Our results improve on this by providing an exact formula for the higher-order term. For online disjoint path allocation, B\"ockenhauer et al.~(ISAAC 2009) gave a lower bound of $\Omega (n/c)$ and an upper bound of $O((n\log c)/c)$ on the advice complexity. We improve on the upper bound by a factor of $\log c$. For the remaining problems, no bounds on their advice complexity were previously known.

\end{abstract}
\section{Introduction}

An online problem is an optimization problem in which the input is divided into small pieces, usually called requests, arriving sequentially. An online algorithm must serve each request without any knowledge of future requests, and the decisions made by the online algorithm are irrevocable. The goal is to minimize or maximize some objective function. 

Traditionally, the quality of an online algorithm is measured by the competitive ratio, which is an analog of the approximation ratio for approximation algorithms: The solution produced by the online algorithm is compared to the solution produced by an optimal offline algorithm, \OPT, which knows the entire request sequence in advance, and only the worst case is considered.

For some online problems, it is impossible to achieve a good competitive ratio. As an example, consider the classical problem of finding a maximum independent set in a graph. Suppose that, at some point, an online algorithm decides to include a vertex $v$ in its solution. It then turns out that all forthcoming vertices in the graph are connected to $v$, but not to each other. Thus, the online algorithm cannot include any of these vertices. On the other hand, \OPT knows the entire graph, and so it rejects $v$ and instead takes all forthcoming vertices. In fact, one can easily show that, even if we allow randomization, no online algorithm for this problem can obtain a competitive ratio better than $\Omega (n)$, where $n$ is the number of vertices in the graph. 

A natural question for online problems, which is not answered by competitive analysis, is the following: \emph{Is there some small amount of information such that, if the online algorithm knew this, then it would be possible to achieve a significantly better competitive ratio?} Our main result is a negative answer to this question for an entire class of hard online problems, including independent set. We prove our main result in the recently introduced advice complexity model. In this model, the online algorithm is provided with $b$ bits of advice about the input. No restrictions are placed on the advice. This means that the advice could potentially encode some knowledge which we would never expect to be in possession of in practice, or the advice could be impossible to compute in any reasonable amount of time. Lower bounds obtained in the advice complexity model are therefore very robust, since they do not rely on any assumptions about the advice. If we know that $b$ bits of advice are necessary to be $c$-competitive, then we know that \emph{any} piece of information which can be encoded using less than $b$ bits will not allow an online algorithm to be $c$-competitive.



In this paper, we use advice complexity to introduce the first complexity class for online problems. The complete problems for this class, one of which is independent set, are very hard in the online setting. We essentially show that for the complete problems in the class, a $c$-competitive online algorithm needs as much advice as is required to explicitly encode a solution of the desired quality. One important feature of our framework is that we introduce an abstract online problem which is complete for the class and well-suited to use as the starting point for reductions. This makes it easy to prove that a large number of online problems are complete for the class and thereby obtain tight bounds on their advice complexity.

\subsection{Advice Complexity}
Advice complexity~\cite{A1, A2, A3, A4} is a quantitative and standardized, i.e., problem independent, way of relaxing the online constraint by providing the algorithm with partial knowledge of the future. 
The main idea of advice complexity is to provide an online algorithm, \ALG, with some advice bits. These bits are provided by a trusted oracle, \ORACLE, which has unlimited computational power and knows the entire request sequence. 

In the first model proposed~\cite{A4}, the advice bits were given as answers (of varying lengths) to questions posed by \ALG. One difficulty with this model is that using at most 1 bit, three different options can be encoded (giving no bits, a 0, or a 1). This problem was addressed by the model proposed in~\cite{A2}, where the oracle is required to send a fixed number of advice bits per request. However, for the problems we consider, one bit per request is enough to guarantee an optimal solution, and so this model is not applicable. Instead, we will use the ``advice-on-tape'' model~\cite{A1}, which allows for a sublinear number of advice bits while avoiding the problem of encoding information in the length of each answer. 
Before the first request arrives, the oracle prepares an \emph{advice tape}, an infinite binary string. The algorithm \ALG may, at any point, read some bits from the advice tape. The {\em advice complexity} of \ALG is the maximum number of bits read by \ALG for any input sequence of at most a given length. 

When advice complexity is combined with competitive analysis, the central question is: How many bits of advice are necessary and sufficient to achieve a given competitive ratio $c$?

\begin{definition}[Competitive ratio \cite{CompRatio1, CompRatio2} and advice complexity \cite{A1, A3}]
The input to an online problem, \P, is a request sequence $\sigma=\langle r_1,\ldots , r_n \rangle$. An \emph{online algorithm with advice}, \ALG, computes the output $y=\langle y_1,\ldots , y_n\rangle$, under the constraint that $y_i$ is computed from $\varphi, r_1,\ldots , r_i$, where $\varphi$ is the content of the advice tape. Each possible output for \P is associated with a \emph{score}. For a request sequence $\sigma$, $\ALG(\sigma)$ $(\OPT(\sigma))$ denotes the score of the output computed by $\ALG$ $(\OPT)$ when serving $\sigma$. 

If \P is a maximization problem, then $\ALG$ is \emph{$c(n)$-competitive} if there exists a constant, $\alpha$, such that, for all $n \in \mathbb{N}$, $$\OPT(\sigma)\leq c(n)\cdot\ALG (\sigma)+\alpha,$$ for all request sequences, $\sigma$, of length at most $n$. 
If \P is a minimization problem, then $\ALG$ is \emph{$c(n)$-competitive} if there exists a constant, $\alpha$, such that, for all $n \in \mathbb{N}$, $$\ALG (\sigma)\leq c(n)\cdot\OPT (\sigma)+\alpha,$$ for all request sequences, $\sigma$, of length at most $n$. 
In both cases, if the inequality holds with $\alpha=0$, we say that $\ALG$ is \emph{strictly $c(n)$-competitive}.

The \emph{advice complexity, $b(n)$, of an algorithm}, \ALG, is the largest number of bits of $\varphi$ read by \ALG over all possible inputs of length at most $n$. 
The {\em advice complexity of a problem}, \P, is a function, $f(n,c)$, $c\geq 1$, such that the smallest possible advice complexity of a strictly $c$-competitive online algorithm for \P is $f(n,c)$.
\end{definition}

In this paper, we only consider deterministic online algorithms (with advice). Note that both the advice read and the competitive ratio may depend on $n$, but, for ease of notation, we often write $b$ and $c$ instead of $b(n)$ and $c(n)$. Also, by this definition, $c\geq 1$, for both minimization and maximization problems. For minimization problems, the score is also called the \emph{cost}, and for maximization problems, the score is also called the \emph{profit}. Furthermore, we use output and \emph{solution} interchangeably.  Lower and upper bounds on the advice complexity have been obtained for many problems, see e.g. \cite{Asushmita, Asetcover, Abip, Aschedule, Ais, Aknapsack, Aec, Asteiner, Avtripartite, Avbipartite, Avpath, Ak-server, Abpj, A1, A2, A3, A4, SG, AListupdate}.

\subsection{String guessing}
In \cite{A2, SG}, the advice complexity of the following {\em string guessing} problem, \oldS, is studied: For each request, which is simply empty and contains no information, the algorithm tries to guess a single bit (or more generally, a character from some finite alphabet). The correct answer is either revealed as soon as the algorithm has made its guess ({\em known history}), or all of the correct answers are revealed together at the very end of the request sequence ({\em unknown history}). The goal is to guess correctly as many bits as possible. 

The problem was first introduced  (under the name generalized matching pennies) in \cite{A2}, where a lower bound for randomized algorithms with advice was given. In \cite{SG}, the lower bound was improved for the case of deterministic algorithms. In fact, the lower bound given in \cite{SG} is tight up to lower-order additive terms. While \oldS is rather uninteresting in the view of traditional competitive analysis, it is very useful in an advice complexity setting. Indeed, it has been shown that the string guessing problem can be reduced to many classical online problems, thereby giving lower bounds on the advice complexity for these problems. This includes bin packing \cite{Abpj}, the $k$-server problem \cite{Asushmita}, list update \cite{AListupdate}, metrical task system \cite{A2}, set cover \cite{SG} and a certain version of maximum clique \cite{SG}.

\subsubsection{Asymmetric string guessing}
In this paper, we introduce a new string guessing problem called {\em asymmetric string guessing, \S}, formally defined in Section~\ref{sectionasg}. The rules are similar to those of the original string guessing problem with an alphabet of size two, but the score function is asymmetric: If the algorithm answers $1$ and the correct answer is $0$, then this counts as a single wrong answer (as in the original problem). On the other hand, if the algorithm answers $0$ and the correct answer is $1$, the solution is deemed infeasible and the algorithm gets an infinite penalty. This asymmetry in the score function forces the algorithm to be very cautious when making its guesses. 

As with the original string guessing problem, \S is not very interesting in the traditional framework of competitive analysis. However, it turns out that \S captures, in a very precise way, the hardness of problems such as online independent set and online vertex cover. 


\subsection{Problems}
Many of the problems that we consider are graph problems, and most of them are studied in the \emph{vertex-arrival model}. In this model, the vertices of an unknown graph are revealed one by one. That is, in each round, a vertex is revealed together with all edges connecting it to previously revealed vertices. For the problems we study in the vertex-arrival model, whenever a vertex, $v$, is revealed, an online algorithm $\ALG$ must (irrevocably) decide if $v$ should be included in its solution or not. Denote by $V_{\ALG}$ the vertices included by $\ALG$ in its solution after all vertices of the input graph have been revealed. The individual graph problems are defined by specifying the set of feasible solutions. The cost (profit) of an infeasible solution is $\infty$ ($-\infty$). 

The problems we consider in the vertex-arrival model are:

\begin{itemize}
\item \VC. A solution is feasible if it is a vertex cover in the input graph. The problem is a minimization problem.

\item \CF. A solution is feasible if the subgraph induced by the vertices in the solution contains a cycle. We assume that the presented graph always contains a cycle. The problem is a minimization problem

\item \DS. A solution is feasible if it is a dominating set in the input graph. The problem is a minimization problem.

\item \IS. A solution is feasible if it is an independent set in the input graph. The problem is a maximization problem.
\end{itemize}

We emphasize that the classical $2$-approximation algorithm for offline vertex cover cannot be used in our online setting, even though the algorithm is greedy. That algorithm greedily covers the edges (by selecting both endpoints) one by one, but this is not possible in the vertex-arrival model.

Apart from the graph problems in the vertex-arrival model mentioned above, we also consider the following online problems. Again, the cost (profit) of an infeasible solution is $\infty$ ($-\infty$).

\begin{itemize}
\item \DPA. A path with $L+1$ vertices $\{v_0,\ldots , v_{L}\}$ is given. Each request $(v_i, v_j)$ is a subpath specified by the two endpoints $v_i$ and $v_j$. A request $(v_i, v_j)$ must immediately be either accepted or rejected. This decision is irrevocable. A solution is feasible if the subpaths that have been accepted do not share any edges. The profit of a feasible solution is the number of accepted paths. The problem is a maximization problem.

\item \SC (set-arrival version). A finite set $U$ known as the \emph{universe} is given. The input is a sequence of $n$ finite subsets of $U$, $(A_1,\ldots , A_n)$, such that $\cup_{1\leq i \leq n}A_i=U$. A subset can be either accepted or rejected. Denote by $S$ the set of indices of the subsets accepted in some solution. The solution is feasible if $\cup_{i\in S}A_i=U$. The cost of a feasible solution is the number of accepted subsets. The problem is a minimization problem.
\end{itemize}


\subsection{Preliminaries}
Throughout the paper, we let $n$ denote the number of requests in the input.

We let $\log$ denote the binary logarithm $\log_2$ and $\ln$ the natural logarithm $\log_e$. 

By a \emph{string} we always mean a bit string. For a string $x\in\{0,1\}^n$, we denote by $\no{x}$ the Hamming weight of $x$ (that is, the number of 1s in $x$) and we define $\nz{x}=n-\no{x}$.  Also, we denote the $i$'th bit of $x$ by $x_i$, so that $x=x_1x_2\ldots x_n$. 

For $n\in\mathbb{N}$, define \mbox{$[n]=\{1,2,\ldots , n\}$}. For a subset $Y\subseteq [n]$, the \emph{characteristic vector} of $Y$ is the string $y=y_1\ldots y_n\in\{0,1\}^n$ such that, for all $i\in[n]$, $y_i=1$ if and only if $i\in Y$. For $x,y\in\{0,1\}^n$, we write $x\sq y$ if $x_i=1\Rightarrow y_i=1$ for all $1\leq i\leq n$. 

If the oracle needs to communicate some integer $m$ to the algorithm, and if the algorithm does not know of any upper bound on $m$, the oracle needs to use a self-delimiting encoding. For instance, the oracle can write $\ceil{\log(m+1)}$ in unary (a string of $1$'s followed by a $0$) before writing $m$ itself in binary. In total, this encoding uses $2\ceil{\log (m+1)}+1=O(\log m)$ bits. Slightly more efficient encodings exist, see e.g.~\cite{Ak-server}.


\subsection{Our contribution}
In Section~\ref{section:ac}, we give lower and upper bounds on the advice complexity of the new asymmetric string guessing problem, \S. The bounds are tight up to an additive term of $O(\log n)$. 
Both upper and lower bounds hold for the competitive ratio as well as the {\em strict} competitive ratio.

More precisely, if $b$ is the number of advice bits necessary and sufficient to achieve a (strict) competitive ratio $c>1$, then we show that 
\begin{equation}
\advice= \log\left(1+\frac{(c-1)^{c-1}}{c^{c}}\right)n \pm\Theta (\log n),\label{maineq}
\end{equation}
where
$$\frac{1}{e\ln 2}\frac{n}{c} \leq \log\left(1+\frac{(c-1)^{c-1}}{c^{c}}\right)n \leq \frac{n}{c}\,.$$
This holds for all variants of the asymmetric string guessing problem (minimization/maximization and known/unknown history). See Figure~\ref{graf} on page~\pageref{graf} for a graphical plot.
For the lower bound, the constant hidden in $\Theta(\log n)$ depends on the additive constant $\alpha$ of the $c$-competitive algorithm. We only consider $c>1$, since in order to be strictly $1$-competitive, an algorithm needs to correctly guess every single bit. It is easy to show that this requires $n$ bits of advice (see e.g.~\cite{SG}). By Remark~\ref{strictremark} in section~\ref{section:ac}, this also gives a lower bound for being $1$-competitive.


In Section~\ref{section:aoc}, we introduce a class, \W , of online problems. The class \W essentially consists of those problems which can be reduced to \S. In particular, for any problem in \W , our upper bound on the advice complexity for \S applies. This is one of the few known examples of a general technique for constructing online algorithms with advice, which works for an entire class of problems.

On the hardness side, we show that several online problems, including \VC, \CF, \DS, \IS, \SC and \DPA are \Wc, that is, they have the same advice complexity as \S. We prove this by providing reductions from \S to each of these problems. The reductions preserve the competitive ratio and only increase the number of advice bits by an additive term of $O(\log n)$. Thus, we obtain bounds on the advice complexity of each of these problems which are essentially tight. Finally, we give a few examples of problems which belong to \W, but are provably not \Wc.
This first complexity class with its many complete problems could be the beginning of a complexity theory for online algorithms.

As a key step in obtaining our results, we establish a connection between the advice complexity of \S and the size of covering designs (a well-studied object from the field of combinatorial designs).

\subsubsection{Discussion of results}
Note that the offline versions of the \Wc problems have very different properties. Finding the shortest cycle in a graph can be done in polynomial time. There is a greedy $2$-approximation algorithm for finding a minimum vertex cover. No $o (\log n)$-approximation algorithm exists for finding a minimum set cover (or a minimum dominating set), unless $\mathsf{P}=\mathsf{NP}$ \cite{Raz}. For any $\varepsilon >0$, no $n^{1-\varepsilon}$-approximation algorithm exists for finding a maximum independent set, unless $\mathsf{ZPP}=\mathsf{NP}$ \cite{Haastad}.
Yet these \Wc problems all have essentially the same high advice complexity. Remarkably, the algorithm presented in this paper for problems in \W is \emph{oblivious} to the input: it ignores the input and uses only the advice to compute the output. Our lower bound proves that for \Wc problems, this oblivious algorithm is optimal. This shows that for \Wc problems, an adversary can reveal the input in such a way that an online algorithm simply cannot deduce any useful information from the previously revealed requests when it has to answer the current request. Thus, even though the \Wc problems are very different in the offline setting with respect to approximation, in the online setting, they become equally hard since an adversary can prevent an online algorithm from using any non-trivial structure of these problems.

Finally, we remark that the bounds (\ref{maineq}) are under the assumption that the number of $1$s in the input string (that is, the size of the optimal solution) is chosen adversarially. In fact, if $t$ denotes the number of $1$s in the input string, we give tight lower and upper bounds on the advice complexity as a function of both $n$, $c$, and $t$. We then obtain (\ref{maineq}) by calculating the value of $t$ which maximizes the advice needed (it turns out that this value is somewhere between $n/(ec)$ and $n/(2c)$). If $t$ is smaller or larger than this value, then our algorithm will use less advice than stated in (\ref{maineq}).

\subsubsection{Comparison with previous results}
The original string guessing problem, \oldS, can be viewed as a maximization problem, the goal being to correctly guess as many of the $n$ bits as possible. Clearly, \OPT always obtains a profit of $n$. With a single bit of advice, an algorithm can achieve a strict competitive ratio of $2$: The advice bit simply indicates whether the algorithm should always guess $0$ or always guess $1$. This is in stark contrast to \S, where linear advice is needed to achieve any constant competitive ratio. On the other hand, for both \oldS and \S, achieving a constant competitive ratio $c<2$ requires linear advice. However, the exact amount of advice required to achieve such a competitive ratio is larger for \S than for \oldS. See Figure~\ref{graf} for a graphical comparison.

The problems \IS and \DPA, which we show to be \Wc, have previously been studied in the context of advice complexity or similar models. 
We present a detailed comparison of our work to these previous results.

In \cite{A1}, among other problems, the advice complexity of \DPA is considered. It is shown that a strictly $c$-competitive algorithm must read at least $\frac{n+2}{2c}-2$ bits of advice. Comparing with our results, we see that this lower bound is asymptotically tight. On the other hand, the authors show that for any $c\geq 2$, there exists a strictly $c$-competitive online algorithm reading at most $b$ bits of advice, where $$b=\min\left\{ n \log \left(\frac{c}{(c-1)^{(c-1)/c}}\right), \frac{n\log n}{c}\right\}+3 \log n + O(1)\,.$$
We remark that $n \log \left(c/(c-1)^{(c-1)/c}\right)\geq (n\log c)/c$, for $c\geq 2$. Thus, this upper bound is a factor of $2 \log c$ away from the lower bound.

In \cite{Magnus}, the problem \IS is studied in a multi-solution model. In this model, an online algorithm is allowed to maintain multiple solutions. The algorithm knows (a priori) the number $n$ of vertices in the input graph. The model is parameterized by a function $r(n)$. Whenever a vertex $v$ is revealed, the algorithm can include $v$ in at most $r(n)$ different solutions (some of which might be new solutions with $v$ as the first vertex). At the end, the algorithm outputs the solution which contains the most vertices.

The multi-solution model is closely related to the advice complexity model. After processing the entire input, an algorithm in the multi-solution model has created at most $n\cdot r(n)$ different solutions (since at most $r(n)$ new solutions can be created in each round). Thus, one can convert a multi-solution algorithm to an algorithm with advice by letting the oracle provide $\log (n\cdot r(n))$ bits of advice indicating which solution to output. In addition, the oracle needs to provide $O(\log n)$ bits of advice in order to let the algorithm learn $n$ (which was given to the multi-solution algorithm for free). On the other hand, an algorithm using $b(n)$ bits of advice can be converted to $2^{b(n)}$ deterministic algorithms. One can then run them in parallel to obtain a multi-solution algorithm with $r(n)=2^{b(n)}$. These simple conversions allow one to translate both upper and lower bounds between the two models almost exactly (up to a lower-order additive term of $O(\log n)$).

It is shown in \cite{Magnus} that for any $c\geq 1$, there is a strictly $c$-competitive algorithm in the multi-solution model if $\ceil{\log r(n)-1} \geq n/c$. This gives a strictly $c$-competitive algorithm reading $\frac{n}{c}+O(\log n)$ bits of advice. On the other hand, it is shown that for any strictly $c$-competitive algorithm in the multi-solution model, it must hold that $c\geq n/(2\log (n\cdot r(n)))$. This implies that any strictly $c$-competitive algorithm with advice must read at least $\frac{n}{2c}-\log n$ bits of advice. Thus, the upper and lower bounds obtained in \cite{Magnus} are asymptotically tight.

Comparing our results to those of \cite{Magnus} and \cite{A1}, we see that we improve on both the lower and upper bounds on the advice complexity of the problems under consideration by giving tight results. For the upper bound on \DPA, the improvement is a factor of $(\log c)/2$. The results of \cite{Magnus} are already asymptotically tight. Our improvement consists of determining the exact coefficient of the higher-order term. Perhaps even more important, obtaining these tight lower and upper bounds on the advice complexity for \IS and \DPA becomes very easy when using our string guessing problem \S. We remark that the reductions we use to show the hardness of these problems reduces instances of \S to instances of \IS (resp.\ \DPA) that are identical to the hard instances used in \cite{Magnus} (resp.~\cite{A1}). What enables us to improve the previous bounds, even though we use the same hard instances, is that we have a detailed analysis of the advice complexity of \S at our disposal.

\subsection{Related work}
The advice complexity of \DPA has also been studied as a function of the length of the path (as opposed to the number of requests), see \cite{A1, DPAL}. 

The advice complexity of \IS on bipartite graphs and on sparse graphs has been determined in \cite{Ais}. It turns out that for these graph classes, even a small amount of advice can be very helpful. For instance, it is shown that a single bit of advice is enough to be $4$-competitive on trees (recall that without advice, it is not possible to be better than $\Omega(n)$-competitive, even on trees).

 It is clear that online maximum clique in the vertex arrival model is essentially equivalent to \IS. In \cite{SG}, the advice complexity of a different version of online maximum clique is studied: The vertices of a graph are revealed as in the vertex-arrival model. Let $V_{\ALG}$ be the set of vertices selected by $\ALG$ and let $C$ be a maximum clique in the subgraph induced by the vertices $V_{\ALG}$. The profit of the solution $V_{\ALG}$ is $\ab{C}^2 / \ab{V_{\ALG}}$. In particular, the algorithm is not required to output a clique, but is instead punished for including too many additional vertices in its output.

The \VC problem and some variations thereof are studied in \cite{Demange}.

The advice complexity of an online set cover problem~\cite{setcover} has been studied in \cite{Asetcover}. However, the version of online set cover that we consider is different and so our results and those of \cite{Asetcover} are incomparable.



\section{Asymmetric String Guessing} \label{sectionasg}
In this section, we formally define the asymmetric string guessing problem and give simple algorithms for the problem. There are four variants of the problem, one for each combination of minimization/maximization and known/unknown history. Collectively, these four problems will be referred to as \S. 

We have deliberately tried to mimic the definition of the string guessing problem \oldS from \cite{SG}. However, for \S, the number, $n$, of requests is not revealed to the online algorithm (as opposed to in \cite{SG}). 
This is only a minor technical detail since it changes the advice complexity by at most $O (\log n)$ bits.

\subsection{The Minimization Version}
We begin by defining the two minimization variants of \S: One in which the output of the algorithm cannot depend on the correctness of previous answers (unknown history), and one in which the algorithm, after each guess, learns the correct answer (known history\footnote{The concept of known history for online problems also appears in \cite{MagnusSzegedy, Magnus} where it is denoted \emph{transparency}.}). We collectively refer to the two minimization problems as \minS.
\begin{definition}
\label{minSudef}
The \emph{minimum asymmetric string guessing problem with unknown history}, \mbox{\minSu}, has input $\langle ?_1,\ldots , ?_n, x\rangle$, where $x\in\{0,1\}^n$, for some $n\in\mathbb{N}$. For $1\leq i \leq n$, round $i$ proceeds as follows:
\begin{enumerate}
\item The algorithm receives request $?_i$ which contains no information.
\item The algorithm answers $y_i$, where $y_i\in\{0,1\}$.
\end{enumerate}
The {\em output} $y=y_1\ldots y_n$ computed by the algorithm is \emph{feasible}, if 
 $x\sq y$. Otherwise, $y$ is \emph{infeasible}. 
The \emph{cost} of a feasible output is $\no{y}$, and the cost of an infeasible output is $\infty$.
The goal is to minimize the cost.
\end{definition}

Thus, each request carries no information. While this may seem artificial, it does capture the hardness of some online problems (see for example Lemma~\ref{dshard}).

\begin{definition}
\label{minSkdef}
The \emph{minimum asymmetric string guessing problem with known history}, \minSk, has input $\langle ?,x_1,\ldots , x_n\rangle$, where $x = x_1 \ldots x_n \in\{0,1\}^n$, for some $n\in\mathbb{N}$. For $1\leq i \leq n$, round $i$ proceeds as follows:
\begin{enumerate}
\item If $i>1$, the algorithm learns the correct answer, $x_{i-1}$, to the request in the previous round.  
\item The algorithm answers $y_i=f(x_1,\ldots , x_{i-1})\in\{0,1\}$, where $f$ is a function defined by the algorithm.

\end{enumerate}
The {\em output} $y=y_1\ldots y_n$ computed by the algorithm is \emph{feasible}, if 
 $x\sq y$. Otherwise, $y$ is \emph{infeasible}. 
The \emph{cost} of a feasible output is $\no{y}$, and the cost of an infeasible output is $\infty$.
The goal is to minimize the cost.
\end{definition}
The string $x$ in either version of \minS will be referred to as the \emph{input string} or the \emph{correct string}. Note that the number of requests in both versions of \minS is $n+1$, since there is a final request that does not require any response from the algorithm. This final request ensures that the entire string $x$ is eventually known. For simplicity, we will measure the advice complexity of \minS as a function of $n$ (this choice is not important as it changes the advice complexity by at most one bit).


Clearly, for any deterministic \minS algorithm which sometimes answers $0$, there exists an input string on which the algorithm gets a cost of $\infty$. However, if an algorithm always answers $1$, the input string could consist solely of $0$s. Thus, no deterministic algorithm can achieve any competitive ratio bounded by a function of $n$. One can easily show that the same holds for any randomized algorithm.


We now give a simple algorithm for \minS which reads $O(n/c)$ bits of advice and achieves a strict competitive ratio of $\ceil{c}$.
\begin{theorem}
For any $c\geq 1$, there is a strictly $\ceil{c}$-competitive algorithm for \minS which reads $\ceil{\frac{n}{c}}+O(\log (n/c))$ bits of advice.
\label{trivialmin}
\end{theorem}
\begin{proof}
We will prove the result for \minSu. Clearly, it then also holds for \minSk.

Let $x=x_1\ldots x_n$ be the input string. The oracle encodes $p=\ceil{n/c}$ in a self-delimiting way, which requires $O(\log (n/c))$ bits of advice. 
For $0\leq j <p$, define $C_j=\{x_i : i \equiv j \pmod p\}$. These $p$ sets partition the input string, and the size of each $C_j$ is at most $\ceil{n/p}$. The oracle writes one bit, $b_j$, for each set $C_j$. If $C_j$ contains only $0$s, $b_j$ is set to $0$. Otherwise, $b_j$ is set to $1$. Thus, in total, the oracle writes $\ceil{n/c}+O(\log (n/c))$ bits of advice to the advice tape.

The algorithm, $\ALG$, learns $p$ and the bits $b_0,\ldots , b_{p-1}$ from the advice tape. In round $i$, $\ALG$ answers with the bit $b_{i\smallmod p}$. We claim that this algorithm is strictly $\ceil{c}$-competitive. It is clear that the algorithm produces a feasible output. Furthermore, if $\ALG$ answers $1$ in round $i$, it must be the case that at least one input bit in $C_{i\smallmod p}$ is $1$. Since the size of each $C_j$ is at most $\ceil{n/p}\leq\ceil{c}$, this implies that $\ALG$ is strictly $\ceil{c}$-competitive.
\end{proof}

\subsection{The Maximization Version}
We also consider \S in a maximization version. One can view this as a dual version of \minS.

\begin{definition}
The \emph{maximum asymmetric string guessing problem with unknown history}, \maxSu, is identical to \minSu, except that the score function is different: The score of a feasible output $y$ is $\nz{y}$, and the score of an infeasible output is $-\infty$. The goal is to maximize the score.
\end{definition}
The maximum asymmetric string guessing problem with {\em known history} is defined similarly:
\begin{definition}
\label{maxSkdef}
The \emph{maximum asymmetric string guessing problem with known history}, \maxSk, is identical to \minSk, except that the score function is different: The score of a feasible output $y$ is $\nz{y}$, and the score of an infeasible output is $-\infty$. The goal is to maximize the score.
\end{definition}

We collectively refer to the two problems as \maxS. 
Similarly, \minSu and \maxSu are collectively called \Su, and \minSk and \maxSk are collectively called \Sk.

An algorithm for \maxS without advice cannot attain any competitive ratio bounded by a function of $n$. If such an algorithm would ever answer $0$ in some round, an adversary would let the correct answer be $1$ and the algorithm's output would be infeasible. On the other hand, answering $1$ in every round gives an output with a profit of zero.

Consider instances of \minS and \maxS with the same correct string $x$. It is clear that the optimal solution is the same for both instances. However, as is usual with dual versions of a problem, they differ with respect to approximation. For example, if half of the bits in $x$ are $1$s, then we get a $2$-competitive solution $y$ for the \minS instance by answering $1$ in each round. However, in \maxS, the profit of the same solution $y$ is zero. Despite this, there is
a similar result to Theorem~\ref{trivialmin} for \maxS. 

\begin{theorem} \label{trivialmax} For any $c\geq 1$, there is a strictly $\ceil{c}$-competitive algorithm for \maxS which reads $\ceil{n/c}+O(\log n)$ bits of advice. \end{theorem}
\begin{proof}
We will prove the result for \maxSu. Clearly, it then also holds for \maxSk.

The oracle partitions the input string $x=x_1\ldots x_n$ into $\ceil{c}$ disjoint blocks, each containing (at most) $\ceil{\frac{n}{c}}$ consecutive bits. Note that there must exist a block where the number of $0$s is at least $\nz{x} / \ceil{c}$. The oracle uses $O(\log n)$ bits to encode the index $i$ in which this block starts and the index $i'$ in which it ends. Furthermore, the oracle writes the string $x_i\ldots x_{i'}$ onto the advice tape, which requires at most $\ceil{\frac{n}{c}}$ bits, since this is the largest possible size of a block. The algorithm learns the string $x_i\ldots x_{i'}$ and answers accordingly in rounds $i$ to $i'$. In all other rounds, the algorithm answers $1$. Since the profit of this output is at least $\nz{x}/\ceil{c}$, it follows that $\ALG$ is strictly $\ceil{c}$-competitive.
\end{proof}

In the following section, we determine the amount of advice an algorithm needs to achieve some competitive ratio $c>1$. It turns out that the algorithms from Theorems \ref{trivialmin} and \ref{trivialmax} use the asymptotically smallest possible number of advice bits, but the coefficient in front of the term $n/c$ can be improved.


\section{Advice Complexity of ASG}
\label{section:ac}
In this section we give upper and lower bounds on the number of advice bits necessary to obtain $c$-competitive ASG algorithms, for some $c>1$.
The bounds are tight up to $O(\log n)$ bits.
For \Su, the gap between the upper and lower bounds stems only from the fact that the advice used for the upper bound includes the number, $n$, of requests and the number, $t$, of 1-bits in the input. 
Since the lower bound is shown to hold even if the algorithm knows $n$ and $t$, this slight gap is to be expected.

The following two observations will be used extensively in the analysis.
\begin{remark}
\label{strictremark}
Suppose that a \minS algorithm, $\ALG$, is $c$-competitive. By definition, there exists a constant, $\alpha$, such that $\ALG(\sigma)\leq c\cdot\OPT(\sigma)+\alpha$. Then, one can construct a new algorithm, $\ALG'$, which is {\em strictly} $c$-competitive and uses $O(\log n)$ additional advice bits as follows: 

Use $O(\log n)$ bits of advice to encode the length $n$ of the input and use $\alpha \cdot \ceil{\log n}=O(\log n)$ bits of advice to encode the index of (at most) $\alpha$ rounds in which $\ALG$ guesses $1$ but where the correct answer is $0$. Clearly, $\ALG'$ can use this additional advice to achieve a strict competitive ratio of $c$. 

This also means that a lower bound of $b$ on the number of advice bits required to be {\em strictly} $c$-competitive implies a lower bound of $b-O(\log n)$ advice bits for being $c$-competitive (where the constant hidden in $O(\log n)$ depends on the additive constant $\alpha$ of the $c$-competitive algorithm). 

The same technique can be used for \maxS.
\end{remark}

\begin{remark}
\label{2balgs}
For a minimization problem, an algorithm, $\ALG$, using $b$ bits of advice can be converted into $2^b$ algorithms, $\ALG_1,\ldots , \ALG_{2^b}$, without advice, one for each possible advice string, such that $\ALG(\sigma)=\min_i \ALG_i(\sigma)$ for any input sequence $\sigma$. The same holds for maximization problems, except that in this case, $\ALG(\sigma)=\max_i \ALG_i(\sigma)$.
\end{remark}

For ASG with unknown history, the output of a deterministic algorithm can depend only on the advice, since no information is revealed to the algorithm through the input. Thus, for \minSu and \maxSu, a deterministic algorithm using $b$ advice bits can produce only $2^b$ different outputs, one for each possible advice string.

\subsection{Using Covering Designs}
In order to determine the advice complexity of \S, we will use some basic results from the theory of combinatorial designs. 
We start with the definition of a covering design.

For any $k \in \mathbb{N}$, a {\em $k$-set} is a set of cardinality $k$.
Let $v\geq k\geq t$ be positive integers. A \emph{(v,k,t)-covering design} is a family of $k$-subsets (called \emph{blocks}) of a $v$-set, $S$, such that any $t$-subset of $S$ is contained in at least one block. The {\em size} of a covering design, $D$, is the number of blocks in $D$. The \emph{covering number}, $C(v,k,t)$, is the smallest possible size of a $(v,k,t)$-covering design. 
Many papers have been devoted to the study of these numbers. See \cite{opac-b1088981} for a survey. The connection to \S is that for inputs to \minS where the number of $1$s is $t$, an $(n, \floor{ct}, t)$-covering design can be used to obtain a strictly $c$-competitive algorithm. 

It is clear that a $(v, k, t)$-covering design always exists. Since a single block has exactly $\binom{k}{t}$ $t$-subsets, and since the total number of $t$-subsets of a set of size $v$ is $\binom{v}{t}$, it follows that ${\binom{v}{t}}/{\binom{k}{t}}\leq C(v,k,t)$. We will make use of the following upper bound on the size of a covering design:
\begin{lemma}[Erd\H{o}s, Spencer \cite{erdos-spencer}] 
\label{erdos}For all natural numbers $v\geq k\geq t$,
\begin{equation*}
\frac{\binom{v}{t}}{\binom{k}{t}}\leq C(v,k,t)\leq \frac{\binom{v}{t}}{\binom{k}{t}}\left(1+\ln\binom{k}{t}\right)
\end{equation*}
\end{lemma} 

We use Lemma~\ref{erdos} to express both the upper and lower bound in terms of (a quotient of) binomial coefficients. This introduces an additional difference of $\log n$ between the stated lower and upper bounds. 

Lemma~\ref{calclower} in Appendix \ref{calcbounds} shows how the bounds we obtain can be approximated by a closed formula, avoiding binomial coefficients. This approximation costs an additional (additive) difference of $O(\log n)$ between the lower and upper bounds. The approximation is in terms of the following function:

\begin{equation*}
\label{bnc}
B(n,c)= \log\left(1+\frac{(c-1)^{c-1}}{c^{c}}\right)n
\end{equation*} 
For $c>1$, we show that $B(n,c)\pm O(\log n)$ bits of advice are necessary and sufficient to achieve a (strict) competitive ratio of $c$, for any version of \S. See Figure \ref{graf} for a graphical view. It can be shown (Lemma~\ref{simplelemma}) that $$\frac{1}{e \ln(2)}\frac{n}{c}\leq B(n,c)\leq \frac{n}{c}\,.$$ In particular, if $c=o(n /\log n)$, we see that $O(\log n)$ becomes a lower-order additive term. Thus, for this range of $c$, we determine exactly the higher-order term in the advice complexity of \S. Since this is the main focus of our paper, we will often refer to $O(\log n)$ as a lower-order additive term. The case where $c=\Omega (n/\log n)$ is treated separately in Section \ref{smalladvice}.

\begin{figure}[!ht]
\begin{center}
\includegraphics[width=0.98\textwidth]{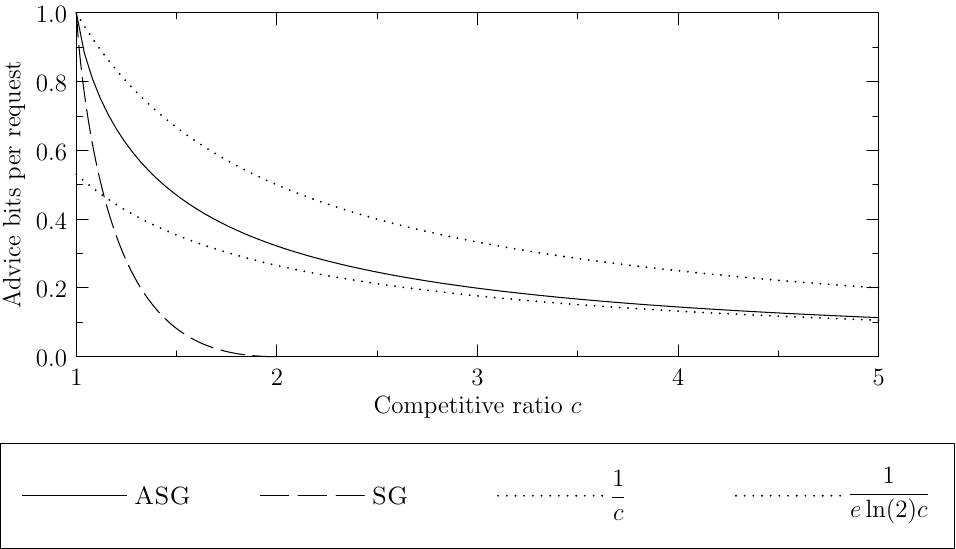}
\end{center}
\caption{The solid line shows the number of advice bits per request which are necessary and sufficient for obtaining a (strict) competitive ratio of $c$ for \S (ignoring lower-order terms). The dashed line shows the same number for the original binary string guessing problem \oldS \cite{SG}. The dotted lines are the functions $1/c$ and $1/(e\ln (2) c)$.}
\label{graf}
\end{figure}

\FloatBarrier

\subsection{Advice Complexity of \minS}
\label{sectionmins}

We first consider \minS with unknown history. 
Clearly, an upper bound for \minSu is also valid for \minSk.
We will show that the covering number $C(v,k,t)$ is very closely related to the advice complexity of \minSu. 

\begin{theorem}
For any $c>1$, there exists a strictly $c$-competitive algorithm for \minS reading $\advice$ bits of advice, where
$$
\advice \leq B(n,c)+O(\log n).
$$
\label{covup}
\end{theorem}
\begin{proof}
We will define an algorithm $\ALG$ and an oracle $\ORACLE$ for \minSu such that $\ALG$ is strictly $c$-competitive and reads at most $\advice$ bits of advice. Clearly, the same algorithm can be used for \minSk. 

Let $x=x_1\ldots x_n$ be an input string to \minSu and set $t=\no{x}$. The oracle $\ORACLE$ writes the value of $n$ to the advice tape using a self-delimiting encoding. Furthermore, the oracle writes the value of $t$ to the advice tape using $\lceil \log n\rceil$ bits (this is possible since $t\leq n$). Thus, this part of the advice uses at most $3\ceil{\log n}+1$ bits in total.

If $\floor{ct}\geq n$, then $\ALG$ will answer $1$ in each round. If $t=0$, $\ALG$ will answer 0 in each round.

If $0<\floor{ct}<n$, then $\ALG$ computes an optimal $(n, \floor{ct}, t)$-covering design as follows: $\ALG$ tries (in lexicographic order, say) all possible sets of $\floor{ct}$-blocks, starting with sets consisting of one block, then two blocks, and so on. For each such set, $\ALG$ can check if it is indeed an $(n, \floor{ct}, t)$-covering design. As soon as a valid covering design, $D$, is found, the algorithm can stop, since $D$ will be a smallest possible $(n, \floor{ct}, t)$-covering design.

Now, $\ORACLE$ picks a $\floor{ct}$-block, $S_y$, from $D$, such that the characteristic vector $y$ of $S_y$ satisfies that $x\sq y$. Note that, since $\ALG$ is deterministic, the oracle knows which covering design $\ALG$ computes and the ordering of the blocks in that design. The oracle then writes the index of $S_y$ on the advice tape. This requires at most $\ceil{\log C(n, \floor{ct}, t)}$ bits of advice. 

$\ALG$ reads the index of the $\floor{ct}$-block $S_y$ from the advice tape and answers $1$ in round $i$ if and only if the element $i$ belongs to $S_y$. Clearly, this will result in $\ALG$ answering $1$ exactly $\floor{ct}$ times and producing a feasible output. It follows that $\ALG$ is strictly $c$-competitive. Furthermore, the number of bits read by $\ALG$ is $$b \leq \left\lceil\log \left(\max_{t\colon \floor{ct}<n}C(n, \floor{ct}, t)\right)\right\rceil+3\lceil \log n\rceil+1\,.$$ 
The theorem now follows from Lemma \ref{calclower}, Inequality~(\ref{UpperIn2}).
\end{proof}

We now give an almost matching lower bound. 

\begin{theorem}
For any $c>1$, a $c$-competitive algorithm $\ALG$ for \minSu must read $\advice$ bits of advice, where
$$\advice \geq B(n,c)-O(\log n)\,.$$
\label{covlow}
\end{theorem}
\begin{proof}
By Remark~\ref{strictremark}, it suffices to prove the lower bound for strictly $c$-competitive algorithms.
Suppose that $\ALG$ is strictly $c$-competitive. Let $\advice$ be the number of advice bits read by $\ALG$ on inputs of length $n$. 
For $0 \leq t \leq n$, let $I_{n,t}$ be the set of input strings of length $n$ with Hamming weight $t$, and let $Y_{n,t}$ be the corresponding set of output strings produced by \ALG.
We will argue that, for each $t$, $0 \leq \lfloor ct \rfloor \leq n$, $Y_{n,t}$ can be converted to an $(n,\floor{ct},t)$-covering design of size at most $2^{b}$.

By Remark~\ref{2balgs}, $\ALG$ can produce at most $2^{b}$ different output strings, one for each possible advice string. Now, for each input string, $x\in I_{n,t}$, there must exist some advice which makes $\ALG$ output a string $y$, where $\no{y}\leq \floor{ct}$ and $x\sq y$. If not, then $\ALG$ is not strictly $c$-competitive. For each possible output $y\in\{0,1\}^n$ computed by $\ALG$, we convert it to the set $S_y\subseteq [n]$ which has $y$ as its characteristic vector. If $\no{y}<\floor{ct}$, we add some arbitrary elements to $S_y$ so that $S_y$ contains exactly $\floor{ct}$ elements. Since $\ALG$ is strictly $c$-competitive, this conversion gives the blocks of an $(n, \floor{ct}, t)$-covering design. The size of this covering design is at most $2^{b}$, since $\ALG$ can produce at most $2^{b}$ different outputs. It follows that $C(n, \floor{ct}, t)\leq 2^{b}$, for all $t$, $0 \leq \lfloor ct \rfloor \leq n$. 
Thus, $$b \geq \log \left(\max_{t\colon\floor{ct}<n}C(n, \floor{ct}, t)\right).$$ 
The theorem now follows from Lemma \ref{calclower}, Inequality~(\ref{LowerIn2}).
\end{proof}

Note that the proof of Theorem \ref{covlow} relies heavily on the unknown history in order to bound the total number of possible outputs. However, Theorem~\ref{lminknown} below states that the lower bound of $B(n,c)-O(\log n)$ also holds for \minSk. 
In order to prove this, we show how an adversary can ensure that revealing the correct answers for previous requests does not give the algorithm too much extra information. The way to ensure this depends on the specific strategy used by the algorithm and oracle at hand, and so the proof is more complicated than that of Theorem \ref{covlow}.

\begin{theorem}
For any $c>1$, a $c$-competitive algorithm for \minSk must read $b$ bits of advice, where
$$\advice \geq B(n,c)-O(\log n)\,.$$
\label{lminknown}
\end{theorem}
\begin{proof}
By Remark~\ref{strictremark}, it suffices to prove the lower bound for strictly $c$-competitive algorithms.
Consider the set, $I_{n,t}$, of input strings of length $n$ and Hamming weight $t$, for some $t$ such that $\lfloor ct \rfloor \leq n$.
Restricting the input set to strings with one particular Hamming weight can only weaken the adversary.

Let $\ALG$ be a strictly $c$-competitive algorithm for \minSk which reads at most $\advice$ bits of advice for any input of length $n$. 
For an advice string $\varphi$, denote by $I_\varphi\subseteq I_{n,t}$ the set of input strings for which $\ALG$ reads the advice $\varphi$.
Since we are considering \minSk, in any round, $\ALG$ may use both the advice string and the information about the correct answer for previous rounds when deciding on an answer for the current round. 

We will prove the lower bound by considering the computation of $\ALG$, when reading the advice $\varphi$, as a game between $\ALG$ and an adversary. This game proceeds according to the rules specified in Definition~\ref{minSkdef}. In particular, at the beginning of round $i$, the adversary reveals the correct answer $x_{i-1}$ for round $i-1$ to $\ALG$. Thus, at the beginning of round $i$, the algorithm knows the first $i-1$ bits, $x_1,\ldots ,x_{i-1}$, of the input string. We say that a string $s\in I_{\varphi}$ is \emph{alive} in round $i$ if $s_j=x_j$ for all $j<i$, and we denote by $I_\varphi^i\subseteq I_{\varphi}$ the set of strings which are alive in round $i$. The adversary must reveal the correct answers in a way that is consistent with $\varphi$. That is, in each round, there must exist at least one string in $I_{\varphi}$ which is alive. 

We first make two simple observations: 
\begin{itemize}
\item Suppose that, in some round $i$, there exists a string $s\in I_\varphi^i$ such that $s_i=1$. Then, $\ALG$ must answer $1$, or else the adversary can choose $s$ as the input string and thereby force $\ALG$ to incur a cost of $\infty$. Thus, we will assume that $\ALG$ always answers $1$ in such rounds.
\item On the other hand, if, in round $i$, all $s \in I_\varphi^i$ have $s_i=0$, then $\ALG$ is free to answer 0. We will assume that $\ALG$ always answers $0$ in such rounds.
\end{itemize}

Assume that, at some point during the computation, \iphi contains exactly $m$ strings and exactly $h$ 1s are still to be revealed.
We let $\mincost(m,h)$ be the largest number such that for every set of $m$ different strings of equal length, each with Hamming weight $h$, 
the adversary can force $\ALG$ to incur a cost of at least $\mincost(m,h)$ when starting for this situation.
In other words, $\mincost(m,h)$ is the minimum number of rounds in which the adversary can force $\ALG$ to answer $1$.

\emph{Claim:} For any $m,h \geq 1$, \begin{equation}\label{claiml} \mincost(m,h)\geq \min\left\{\minbits \colon m\leq\binom{\minbits}{h}\right\}.\end{equation}

Before proving the claim, we will show how it implies the theorem. For any $t$,  $0 \leq \floor{ct} < n$, there are $\binom{n}{t}$ possible input strings of length $n$ and Hamming weight $t$. By the pigeonhole principle, there must exist an advice string $\varphi'$ such that $\ab{I_{\varphi'}}\geq\binom{n}{t}/2^{\advice}$. Now, if $m=\ab{I_{\varphi'}}>\binom{\floor{ct}}{t}$, then by (\ref{claiml}), $\mincost(m,t)\geq \min \{\minbits \colon \binom{\floor{ct}}{t} < \binom{\minbits}{t} \} = \floor{ct}+1$. This contradicts the fact that $\ALG$ is strictly $c$-competitive. Thus, it must hold that $\ab{I_{\varphi'}}\leq\binom{\floor{ct}}{t}$. 
Combining the two inequalities involving $\ab{I_{\varphi'}}$, we get
$$\binom{\floor{ct}}{t} \geq \ab{I_{\varphi'}} \geq \frac{\binom{n}{t}}{2^{\advice}} \; \Rightarrow \; 2^{\advice} \geq \frac{\binom{n}{t}}{\binom{\floor{ct}}{t}}$$
Since this holds for all values of $t$, we obtain the lower bound
$$\advice \geq\log \left(\max_{t\colon\floor{ct}< n}\frac{\binom{n}{t}}{\binom{\floor{ct}}{t}}\right).$$
The theorem then follows from Lemma~\ref{calclower} and Inequalities~(\ref{UpperIn1}) and (\ref{LowerIn2}).

\paragraph{Proof of claim:} Fix $1\leq i\leq n$ and assume that, at the beginning of round $i$, there are $m$ strings alive, all of which still have exactly $h$ 1's to be revealed. The rest of the proof is by induction on $m$ and $h$. 

For the {\em base case}, suppose first that $h=1$. Then, for each of the $m$ strings, $s^1,\ldots , s^m\in I_\varphi^i$, there is exactly one index, $i_1$,\ldots , $i_m$, such that $s^1_{i_1}=\cdots = s^m_{i_m} = 1$. Since all strings in $I_\varphi^i$ must be different, it follows that $i_j\neq i_k$ for $j\neq k$. Without loss of generality, assume that $i_1<i_2<\cdots < i_m$. In rounds $i_1,\ldots , i_{m-1}$, the adversary chooses the correct answer to be $0$, while $\ALG$ is forced to answer $1$ in each of these rounds. Finally, in round $i_m$, the adversary reveals the correct answer to be $1$ (and hence the input string must be $s^m$). In total, $\ALG$ incurs a cost of $m$, which shows that $\mincost(m,1)= m$ for all $m\geq 1$. 

Assume now that $m=1$. It is clear that $\mincost(m,h)\geq h$ for all values of $h$. In particular, $\mincost(1,h)=h$. This finishes the base case.

For the {\em inductive step}, fix integers $m,h\geq 2$. Assume that the formula is true for all $(i,j)$ such that $j\leq h-1$ or such that $j=h$ and $i\leq m-1$. We will show that the formula is also true for $(m, h)$. 

Consider the strings $s^1,\ldots , s^m\in I_\varphi^i$ alive at the beginning of round $i$. We partition $I_{\varphi}^i$ into two sets, $S_0=\{s^j \colon s^j_i=0\}$ and $S_1= \{s^j \colon s^j_i=1\}$, and let $\mzero=\ab{S_0}$ and $\mone=\ab{S_1}$. 
Recall that if all sequences $s \in I_{\varphi}^i$ have $s_i=0$, we assume that \ALG answers 0, leaving $m$ and $h$ unchanged.
Thus, we may safely ignore such rounds and assume that $\mzero<m$. We let 
\begin{align*}
\dstart & = \min\left\{\minbits' \colon m\leq\binom{\minbits'}{h}\right\},\\
\dzero  & = \min\left\{\minbits' \colon \mzero\leq\binom{\minbits'}{h}\right\}, \text{ and }\\
\done   & = \min\left\{\minbits' \colon \mone\leq\binom{\minbits'}{h-1}\right\}.
\end{align*}

If $\done+1 \geq \dstart$, then the adversary chooses $1$ as the correct answer in round $i$. By the induction hypothesis, $\mincost(\mone, h-1)\geq \done$. Together with the fact that $\ALG$ is forced to answer $1$ in round $i$, this shows that the adversary can force $\ALG$ to incur a cost of at least $\mincost(\mone, h-1)+1\geq \done+1\geq \dstart$.

On the other hand, if $\done+1<\dstart$, the adversary chooses $0$ as the correct answer in round $i$. Note that this implies that each string alive in round $i+1$ still has exactly $h$ 1's to be revealed. We must have $\done\leq \dstart-2$ since $\done$ and $\dstart$ are both integers. Moreover, by definition of $\dstart$, it holds that $m>\binom{\dstart-1}{h}$. Thus, we get the following lower bound on $\mzero$:
\begin{align*}
\mzero&=m-\mone\\
&> \binom{\dstart-1}{h}-\binom{\done}{h-1}\\
&\geq\binom{\dstart-1}{h}-\binom{\dstart-2}{h-1}, \text{ since $\binom{a}{b}$ is increasing in $a$}\\
&=\binom{\dstart-2}{h}, \text{ by Pascal's Identity}.
\end{align*}
This lower bound on \mzero shows that $\dzero>\dstart-2$, and hence $\dzero\geq \dstart-1$. Combining this with the induction hypothesis gives $\mincost(\mzero,h)\geq \dzero \geq \dstart-1$. Since $m_1\geq 1$, $\ALG$ is forced to answer $1$ in round $i$, so the adversary can make $\ALG$ incur a cost of at least $\mincost(\mzero,h)+1\geq \dstart$. 
\end{proof}


\subsection{Advice Complexity of \maxS}
\label{sectionmaxs}
In this section, we will show that the advice complexity of \maxS is the same as that of \minS, up to a lower-order additive term of $O(\log n)$. 
We use the same techniques as in Section~\ref{sectionmins}.

As noted before, the difficulty of computing a $c$-competitive solution for a specific input string is not the same for \minS and \maxS. The key point is that computing a $c$-competitive solution for \maxS, on input strings with $u$ $0$'s, is roughly as difficult as computing a $c$-competitive solution for \minS, on input strings with $\ceil{u/c}$ $1$'s. 

We show that the proofs of Theorems~\ref{covup}--\ref{lminknown} can easily be modified to give upper and lower bounds on the advice complexity of \maxS. These bounds within the proofs look slightly different from the ones obtained for \minS, but we show in Lemmas~\ref{binomequal} and \ref{calclowerM} that they differ from $B(n,c)$ by at most an additive term of $O(\log n)$. 
\begin{theorem} \label{maxsgalg}
For any $c>1$, there exists a strictly $c$-competitive online algorithm for \maxS reading $\advice$ bits of advice, where
$$
\advice \leq  B(n,c)+O(\log n).
$$
\end{theorem}
\begin{proof}
We will define an algorithm $\ALG$ and an oracle $\ORACLE$ for \maxSu such that $\ALG$ is strictly $c$-competitive and reads at most $\advice$ bits of advice. Clearly, the same algorithm can be used for \maxSk. 

As in the proof of Theorem \ref{covup}, we note that, for any integers $n,u$ where $0<u<n$, the algorithm $\ALG$ can compute an optimal $(n,n-\ceil{u/c},n-u)$-covering design deterministically. 

Let $x=x_1\ldots x_n$ be an input string to \maxSu and set $u=\nz{x}$. The oracle $\ORACLE$ writes the values of $n$ and $u$ to the advice tape using at most $3\ceil{\log n}+1$ bits in total.

If $0<u<n$, then $\ORACLE$ picks an $(n-\ceil{u/c})$-block, $S_y$, from the optimal $(n,n-\ceil{u/c},n-u)$-covering design, as computed by $\ALG$, such that the characteristic vector $y$ of $S_y$ satisfies that $x\sq y$. The oracle writes the index of $S_y$ on the advice tape. This requires at most $\ceil{\log C(n,n-\ceil{u/c},n-u)}$ bits of advice. 

The algorithm, $\ALG$, first reads the values of $n$ and $u$ from the advice tape. If $u=0$, then $\ALG$ will answer $1$ in each round, and if $u=n$, then $\ALG$ will answer $0$ in each round. If $0<u<n$, then $\ALG$ will read the index of the ($n-\ceil{u/c}$)-block $S_y$ from the advice tape. $\ALG$ will answer $1$ in round $i$ if and only if the element $i$ belongs to the given block. Clearly, this will result in $\ALG$ answering $0$ exactly $n-(n-\ceil{u/c})=\ceil{u/c}$ times and producing a feasible output. It follows that $\ALG$ will be strictly $c$-competitive. Furthermore, the number of bits read by $\ALG$ is $$b \leq \left\lceil\log \left(\max_{u \colon 0<u<n}C\left(n,n-\left\lceil\frac{u}{c}\right\rceil,n-u\right)\right)\right\rceil+3\lceil \log n\rceil+1\,.$$
The theorem now follows from Lemma \ref{calclowerM}.
\end{proof}

\begin{theorem}
\label{upacklow}
For any $c>1$, a $c$-competitive algorithm $\ALG$ for \maxSu must read  $b$ bits of advice, where
$$b \geq B(n,c)-O(\log n)\,.$$
\end{theorem}

\begin{proof}
By Remark~\ref{strictremark}, it suffices to prove the lower bound for strictly $c$-competitive algorithms. Suppose that $\ALG$ is strictly $c$-competitive. Let $\advice$ be the number of advice bits read by $\ALG$ on inputs of length $n$. 
For $0 \leq u \leq n$, let $I_{n,u}$ be the set of input strings $x$ of length $n$ with $\nz{x}=u$, and let $Y_{n,u}$ be the corresponding set of output strings produced by \ALG.
We will argue that, for each $u$, $Y_{n,u}$ can be converted to an $(n,n-\ceil{u/c},n-u)$-covering design of size at most $2^{b}$.

By Remark~\ref{2balgs}, $\ALG$ can produce at most $2^{b}$ different output strings, one for each possible advice string. Now, for each input string, $x=x_1 \ldots x_n$ with $\nz{x}=u$ (and, hence, $\no{x}=n-u$), there must exist some advice which makes $\ALG$ output a string $y=y_1 \ldots y_n$ where $\nz{y}\geq \ceil{u/c}$ (and, hence, $\no{y}\leq n-\ceil{u/c}$) and $x\sq y$. If not, then $\ALG$ is not strictly $c$-competitive. For each possible output $y\in\{0,1\}^n$ computed by $\ALG$, we convert it to the set $S_y\subseteq [n]$ which has $y$ as its characteristic vector. If $\no{y}<n-\ceil{u/c}$, we add some arbitrary elements to $S_y$ so that $S_y$ contains exactly $n-\ceil{u/c}$ elements. Since $\ALG$ is strictly $c$-competitive, this conversion gives the blocks of an $(n, n-\ceil{u/c}, n-u)$-covering design. The size of this covering design is at most $2^{b}$, since $\ALG$ can produce at most $2^{b}$ different outputs. It follows that $C(n, n-\ceil{u/c}, n-u)\leq 2^{b}$, for all $u$. 
Thus, $$b \geq \log \left(\max_{u \colon 0<u<n}C\left(n, n-\left\lceil \frac{u}{c} \right\rceil, n-u\right)\right).$$ The theorem now follows from Lemma \ref{calclowerM}.
\end{proof}

As was the case for \minS, the lower bound for \maxSu also holds for \maxSk.
\begin{theorem}
\label{kpacklow}
For any $c>1$, a $c$-competitive algorithm $\ALG$ for \maxSk must read at least $b$ bits of advice, where $$b\geq B(n,c)-O(\log n)\,.$$ 
\label{packlow}
\end{theorem}
\begin{proof}
By Remark~\ref{strictremark}, it suffices to prove the lower bound for strictly $c$-competitive algorithms.

Consider input strings, $x$, of length $n$ and such that $\nz{x}=u$. Let $t=\no{x}=n-u$. 
We reuse the notation from the proof of Theorem~\ref{lminknown} and let $I_\varphi \subseteq I_{n,t}$ denote the set of strings for which \ALG reads the advice string $\varphi$. 

Suppose there exists some advice string $\varphi'$ such that $m = \ab{I_{\varphi'}} > \binom{n-\ceil{\frac{u}{c}}}{t}$. Since Inequality~(\ref{claiml}) from the proof of Theorem~\ref{lminknown} holds for \maxS too, we get that $L_1(m,t)\geq  n-\ceil{\frac{u}{c}}+1$. But this means that there exists an input $x\in I_{\varphi'}$, with $\no{x}=t$, such that $\ALG$ must answer 1 at least $n-\ceil{\frac{u}{c}}+1$ times. In other words, for the output $y$, computed by $\ALG$ on input $x$, it holds that $\nz{y}\leq n-(n-\ceil{\frac{u}{c}}+1)\leq \ceil{\frac{u}{c}}-1$. Since $\nz{x}=u$, this contradicts the fact that $\ALG$ is strictly $c$-competitive.

Since there are $\binom{n}{u}$ possible input strings $x$ such that $\nz{x}=u$, and since the above was shown to hold for all choices of $u$, we get the lower bound $$b \geq \log \left(\max_{u\colon0<u<n}\frac{\binom{n}{u}}{\binom{n-\ceil{\frac{u}{c}}}{n-u}}\right).$$ 
The theorem now follows from Lemma \ref{calclowerM}.
\end{proof}

\subsection{Advice Complexity of \S when $c=\Omega(n/\log n)$} \label{smalladvice}
Throughout the paper, we mostly ignore additive terms of $O(\log n)$ in the advice complexity. However, in this section, we will consider the advice complexity of \S when the number of advice bits read is at most logarithmic. Surprisingly, it turns out that the advice complexity of \minS and \maxS is different in this case.

Recall that, by Theorem \ref{maxsgalg} (or Theorem \ref{trivialmax}), using $O(\log n)$ bits of advice, an algorithm for \maxS can achieve a competitive ratio of $\frac{n}{\log n}$. The following theorem shows that there is a ``phase-transition'' in the advice complexity, in the sense that using less than $\log n$ bits of advice is no better than using no advice at all. We remark that Theorem \ref{smallamax} and its proof are essentially equivalent to a previous result of Halld\'orsson et al.~\cite{Magnus} on \IS in the multi-solution model.

\begin{theorem}[cf.\cite{Magnus}]
\label{smallamax}
Let $\ALG$ be an algorithm for \maxS reading $b<\floor{\log n}$ bits of advice. Then, the competitive ratio of $\ALG$ is not bounded by a function of $n$. This is true even if $\ALG$ knows $n$ in advance.
\end{theorem}
\begin{proof}
We will prove the result for \maxSk. Clearly, it then also holds for \maxSu.

By Remark~\ref{2balgs}, we can convert $\ALG$ to $m=2^{b}$ online algorithms without advice. Denote the algorithms by $\ALG_1,\ldots , \ALG_m$. Since $b< \floor{\log n}$, it follows that $m\leq n/2$. We claim that the adversary can construct an input string $x=x_1 \ldots x_n$ for \maxSk such that the following holds: For each $1\leq j\leq m$, the output of $\ALG_j$ is either infeasible or contains only $1$s. Furthermore, $x$ can be constructed such that $\nz{x}\geq \frac{n}{2}$. 

We now show how the adversary may achieve this. For $1\leq i \leq n$, the adversary decides the value of $x_i$ as follows: If there is some algorithm, $\ALG_j$, which answers $0$ in round $i$ and $\ALG_j$
answers $1$ in all rounds before round $i$, the adversary lets $x_i=1$. In all other cases, the adversary lets $x_i=0$. It follows that if an algorithm $\ALG_j$ ever answers $0$, its output will be infeasible. Furthermore, the number of $1$'s in the input string constructed by the adversary is at most $n/2$, since $m\leq n/2$. Thus, the profit of $\OPT$ on this input is at least $n/2$, while the profit of $\ALG$ is at most $0$.
\end{proof}

For \minS, the algorithm from Theorem \ref{trivialmin} achieves a competitive ratio of $\ceil{c}$ and uses $O(n/c)$ bits of advice, for any $c>1$. In particular, it is possible to achieve a competitive ratio of e.g. $O(n/(\log \log n))$ using $O(\log \log n)$ bits of advice, which we have just shown is not possible for \maxS. The following theorem shows that no strictly $c$-competitive algorithm for \minS can use less than $\Omega (n/c)$ bits of advice, even if $n/c=o(\log n)$.

\begin{theorem}
\label{smallamin}
For any $c>1$, on inputs of length $n$, a strictly $\ceil{c}$-competitive algorithm $\ALG$ for \minS must read at least $b=\Omega (n/c)$ bits of advice.
\end{theorem}
\begin{proof}
We will prove the result for \minSk. Clearly, it then also holds for \minSu.

Suppose that $\ALG$ is strictly $\ceil{c}$-competitive. Since $\floor{\ceil{c}t}=\ceil{c}t$, it follows from the proof of Theorem \ref{lminknown} that $\ALG$ must read at least $b$ bits of advice, where
$$b \geq\log \left(\max_{t\colon \ceil{c}t< n}\frac{\binom{n}{t}}{\binom{\ceil{c}t}{t}}\right).$$ By Lemma \ref{finallemma}, this implies that $b=\Omega(n/c)$.
\end{proof}



\section{The Complexity Class \W} \label{sectionclass}
\label{section:aoc}
In this section, we define a class, \W, and show that for each problem, \P, in \W, the advice complexity of \P is at most that of \S.

\begin{definition} \label{sgeasydef}\label{wdef}
A problem, \P, is in \W \emph{(Asymmetric Online Covering)} if it can be defined as follows:
The input to an instance of \P consists of a sequence of $n$ requests, $\sigma= \langle r_1, \ldots, r_n \rangle$, and possibly one final dummy request. An algorithm for \P computes a binary output string, $y=y_1 \ldots y_n\in\{0,1\}^n$, where $y_i=f(r_1, \ldots , r_i)$ for some function $f$. 

For minimization (maximization) problems, the score function, $s$, maps a pair, $(\sigma,y)$, of input and output to a cost (profit) in $\mathbb{N} \cup \{\infty \}$ $(\mathbb{N} \cup \{-\infty \})$.
 For an input, $\sigma$, and an output, $y$, $y$ is \emph{feasible} if $s(\sigma,y) \in \mathbb{N}$. Otherwise, $y$ is \emph{infeasible}. There must exist at least one feasible output. Let $S_{\min}(\sigma)$ $(S_{\max}(\sigma))$ be the set of those outputs that minimize (maximize) $s$ for a given input $\sigma$.

If \P is a minimization problem, then for every input, $\sigma$, the following must hold:
\begin{enumerate}
\item For a feasible output, $y$, $s(\sigma,y)=\no{y}$.
\item An output, $y$, is feasible if 
   there exists a $y'\in S_{\min}(\sigma)$ such that $y'\sq y$.\\
  If there is no such $y'$, the output may or may not be feasible.
\end{enumerate}

If \P is a maximization problem, then for every input, $\sigma$, the following must hold:
\begin{enumerate}
\item For a feasible output, $y$, $s(\sigma,y)=\nz{y}$.
\item An output, $y$, is feasible if 
   there exists a $y'\in S_{\max}(\sigma)$ such that $y'\sq y$.\\
  If there is no such $y'$, the output may or may not be feasible.
\end{enumerate}
\end{definition}

The dummy request is a request that does not require an answer and is not counted when we count the number of requests. Most of the problems that we consider will not have such a dummy request, but it is necessary to make sure that \S belongs to \W.


The input, $\sigma$, to a problem \P in \W can contain any kind of information. However, for each request, an algorithm for \P only needs to make a binary decision. If the problem is a minimization problem,
it is useful to think of answering $1$ as accepting the request and answering $0$ as rejecting the request (e.g. vertices in a vertex cover). 
The output is guaranteed to be feasible if the accepted requests are a superset of the requests accepted in an optimal solution (they ``cover'' the optimal solution).

If the problem is a maximization problem, it is useful to think of answering $0$ as accepting the request and answering $1$ as rejecting the request (e.g. vertices in an independent set).
The output is guaranteed to be feasible if the accepted requests are a subset of the requests accepted in a optimal solution.

Note that outputs for problems in \W may have a score of $\pm\infty$. This is used to model that the output is infeasible (e.g. not a vertex cover/independent set).


We now show that our \Su algorithm based on covering designs works for every problem in \W. This gives an upper bound on the advice complexity for all problems in \W.
\begin{theorem} \label{algsgeasy}
Let \P be a problem in \W. There exists a strictly $c$-competitive online algorithm for \P reading $\advice$ bits of advice, where
\begin{align*}
\advice \leq B(n,c)+O(\log n).
\end{align*}
\end{theorem}
\begin{proof}
We first assume that \P is a minimization problem. Let \ALG be a strictly $c$-competitive \minSu algorithm reading at most $b$ bits of advice provided by an oracle \ORACLE. By Theorem \ref{covup}, such an algorithm exists. We will define a \P algorithm, \ALGp, together with an oracle \ORACLEp, that is strictly $c$-competitive and reads at most $\advice$ bits of advice. 

For a given input, $\sigma$, to \P, the oracle \ORACLEp starts by computing an $x$ such that $x\in S_{\min}(\sigma)$. This is always possible since by the definition of \W, such an $x$ always exists, and \ORACLEp has unlimited computational power.
Let $\varphi$ be the advice that \ORACLE would write to the advice tape if $x$ were the input string in an instance of \minSu. \ORACLEp writes $\varphi$ to the advice tape. From here, \ALGp behaves as \ALG would do when reading $\varphi$ (in particular, \ALGp ignores any possible information contained in $\sigma$) and computes the output $y$. Since \ALG is strictly $c$-competitive for \minSu, we know that $x\sq y$ and that $\no{y}\leq c\no{x}$. 
Since \P is in \W, this implies that $y$ is feasible (with respect to the input $\sigma$) and that $s(\sigma,y)\leq c \no{x}=c \cdot \OPT(\sigma)$. 

Similarly, one can reduce a maximization problem to \maxSu and apply Theorem \ref{maxsgalg}.
\end{proof}

Showing that a problem, \P, belongs to \W immediately gives an upper bound on the advice complexity of \P. For all variants of \S, we know that this upper bound is tight up to an additive $O(\log n)$ term. This leads us to the following definition of completeness.

\begin{definition}
\label{completedef}
A problem, \P, is \emph{\Wc} if 
\begin{itemize}
\item \P belongs to \W and 
\item for all $c>1$, any $c$-competitive algorithm for \P must read at least $b$  bits of advice, where
$$
\advice \geq B(n,c)-O(\log n).
$$
\end{itemize}
\end{definition}

Thus, the advice complexity of an \Wc problem must be identical to the upper bound from Theorem~\ref{algsgeasy}, up to a lower-order additive term of $O(\log n)$. By Definitions~\ref{minSudef}--\ref{maxSkdef} combined with Theorems~\ref{covlow}--\ref{lminknown} and \ref{upacklow}--\ref{packlow}, all of \minSu, \minSk, \maxSu and \maxSk are \Wc. 

When we show that some problem, \P, is \Wc, we usually do this by giving a reduction from a known \Wc problem to \P, preserving the competitive ratio and increasing the number of advice bits by at most $O(\log n)$. \Sk is especially well-suited as a starting point for such reductions. 

We allow for an additional $O(\log n)$ bits of advice in Definition~\ref{completedef} in order to be able to use the reduction between the strict and non-strict competitive ratios as explained in Remark~\ref{strictremark} and in order to encode some natural parameters of the problem, such as the input length or the score of an optimal solution. For most values of $c$, it seems reasonable to allow these additional advice bits. However, it does mean that for $c=\Omega (n/\log n)$, the requirement in the definition of \Wc is vacuously true. We refer to Section \ref{smalladvice} for a discussion of the advice complexity for this range of competitive ratio.



\subsection{\Wc Minimization Problems}\label{appmins}
In this section, we show that several online problems are \Wc, starting with \VC. See the introduction for the definition of the problems

\subsubsection{Online Vertex Cover.}

\begin{lemma} \label{VCupperlemma}
\VC is in \W.
\end{lemma}
\begin{proof}
We need to verify the conditions in Definition~\ref{wdef}. 

Recall that an input $\sigma= \langle r_1, \ldots, r_n \rangle$ for \VC is a sequence of requests, where each request is a vertex along with the edges connecting it to previously requested vertices. There is no dummy request at the end.
For each request, $r_i$, an algorithm makes a binary choice, $y_i$: It either includes the vertex into its solution ($y_i=1$) or not ($y_i=0$). 

The cost of an infeasible solution is $\infty$.
A solution $y=y_1 \ldots y_n$ for \VC is feasible if the vertices included in the solution form a vertex cover in the input graph. Clearly, there is always at least one feasible solution, since taking all the vertices will give a vertex cover. 

Thus, \VC has the right form.
Finally, we verify that conditions 1 and 2 are also satisfied: Condition 1 is satisfied since the cost of a feasible solution is the number of vertices in the solution and condition 2 is satisfied since a superset of a vertex cover is also a vertex cover. 
\end{proof}

We now show a hardness result for \VC. In our reduction, we make use of the following graph construction. The same construction will also be used later on for other problems. We remark that this graph construction is identical to the one used in \cite{Magnus} for showing lower bounds for \IS in the multi-solution model.

\begin{definition}[cf. \cite{Magnus}]
\label{G}
For any string $x=x_1\ldots x_n\in\{0,1\}^n$, define $G_x=(V,E)$ as follows:
\begin{align*}
V&=\{v_1,\ldots , v_n\},\\
E&=\{(v_i, v_j) \colon \text{$x_i=1$ and $i<j$}\}.
\end{align*} 
Furthermore, let $V_0=\{v_i \colon x_i=0\}$ and $V_1=\{v_i \colon x_i =1\}$.
\end{definition}

For a string $x\in\{0,1\}^n$, the graph $G_x$ from Definition~\ref{G} is a \emph{split graph}: The vertex set $V$ can be partitioned into $V_0$ and $V_1$ such that $V_0$ is an independent set of size $\nz{x}$ and $V_1$ is a clique of size $\no{x}$.

\begin{figure}[h]
\centering
\begin{tikzpicture}[-,>=stealth',minimum size=0.75cm,node distance=2cm,
  thick,main node/.style={circle,fill=blue!20,minimum size=0.75cm,draw,font=\sffamily\Large\bfseries}]
\selectcolormodel{gray}

  \node[main node] (1) {{\small 0}};
  \node[main node] (2) [right=1 cm of 1] {{\small 1}};
  \node[main node] (3) [right=1 cm of 2] {{\small 1}};
  \node[main node] (4) [right=1 cm of 3] {{\small 0}};
  \node[main node] (5) [right=1 cm of 4] {{\small 1}};
  \node[main node] (6) [right=1 cm of 5] {{\small 0}};

  \path[every node/.style={font=\sffamily\small}]
	(2) edge (3)
	(2) edge [bend left] (4)
	(2) edge [bend left] (5)
	(2) edge [bend left] (6)
	(3) edge (4)
	(3) edge [bend right] (5)
	(3) edge [bend right] (6)
	(5) edge (6)

	;
\end{tikzpicture}
\caption{$G_{011010}$}
\end{figure}
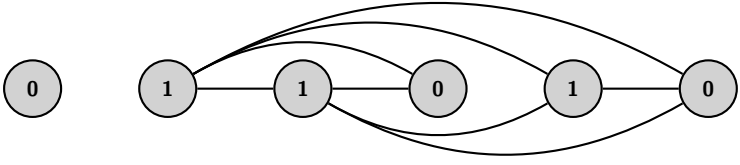

\begin{lemma}
\label{vckh}
If there is a $c$-competitive algorithm reading $\advice$ bits for \VC, then there is a $c$-competitive algorithm reading $\advice + O(\log n)$ bits for \minSk.
\end{lemma}
\begin{proof}
Let $\ALG$ be a $c$-competitive algorithm for \VC reading at most $\advice$ bits of advice. By definition, there exists a constant $\alpha$ such that $\ALG(\sigma)\leq c\cdot\OPT(\sigma)+\alpha$ for any input sequence $\sigma$. We will define an algorithm, $\ALGp$, and an oracle, \ORACLEp, for \minSk such that $\ALGp$ is $c$-competitive (with the same additive constant) and reads at most $\advice+ O(\log n)$ bits of advice.

For $x=x_1\ldots x_n$ an input string to \minSk, consider the input instance to \VC $G_x=(V,E)$ defined in Definition~\ref{G} where the vertices are requested in the order $\langle v_1, \ldots, v_n\rangle$. We say that a vertex in $V_0$ is \emph{bad} and that a vertex in $V_1$ is \emph{good}. Note that $V_1\setminus\{v_n\}$ is a minimum vertex cover of $G_x$. Also, if an algorithm rejects a good vertex $v_i$, then it must accept all later vertices $v_j$ (where $i<j\leq n$) in order to cover the edges $(v_i, v_j)$. In particular, since the good vertices form a clique, no algorithm can reject more than one good vertex.

Let $\varphi$ be the advice read by $\ALG$, and let $V_{\ALG}$ be the vertices chosen by $\ALG$. Since $\ALG$ is $c$-competitive, we know that $V_{\ALG}$ must be a vertex cover of size at most $c\ab{V_1\setminus\{v_n\}}+\alpha\leq c\ab{V_1}+\alpha$. 

We now define $\ALGp$ and $\ORACLEp$. As usual, $y$ denotes the output computed by $\ALGp$. We consider three cases. The first two bits of the advice tape will be used to tell \ALGp which one of the three cases we are in. 

\emph{Case 1: $\ALG$ accepts all good vertices in $G_x$, i.e., $V_1\subseteq V_\ALG$.}
 The oracle \ORACLEp writes the advice $\varphi$ to the advice tape. When $\ALGp$ receives request $i$, it considers what $\ALG$ does when the vertex $v_i$ in $G_x$ is revealed: \ALGp answers $1$ if \ALG accepts $v_i$ and $0$ otherwise. Note that it is possible for $\ALGp$ to simulate $\ALG$ since, at the beginning of round $i$, $\ALGp$ knows $x_1\ldots x_{i-1}$. In particular, $\ALGp$ knows which edges to reveal to $\ALG$ along with the vertex $v_i$ in $G_x$. Together with access to the advice $\varphi$ read by $\ALG$, this allows $\ALGp$ to simulate $\ALG$. Since $V_1\subseteq V_\ALG$, we get that $x\sq y$. Furthermore, since $\ab{V_\ALG}\leq c\ab{V_1}+\alpha$, we also get that $\no{y}\leq c\no{x}+\alpha$.



\emph{Case 2a: \ALG rejects a good vertex, $v_i$, and accepts a bad vertex, $v_j$.} In this case, the oracle \ORACLEp writes the indices of $i$ and $j$ in a self-delimiting way, followed by $\varphi$, to the advice tape. $\ALGp$ simulates $\ALG$ as before and answers accordingly, except that it answers $1$ in round $i$ and $0$ in round $j$. This ensures that $x\sq y$. Furthermore, $\no{y}=\ab{V_\ALG}\leq c\ab{V_1}+\alpha= c\no{x}+\alpha$.

\emph{Case 2b: \ALG rejects a good vertex, $v_i$, and all bad vertices.} In this case, $V_{\ALG}=V_1\setminus\{v_i\}$. The oracle \ORACLEp writes 
the value of $i$ to the advice tape in a self-delimiting way, followed by $\varphi$. Again, \ALGp simulates \ALG, but it answers $1$ in round $i$. 
Thus, $x=y$, meaning that $x \sq y$ and $y$ is optimal.

In all cases, $\ALGp$ computes an output $y$ such that $x\sq y$ and $\no{y}\leq c\no{x}+\alpha$. 
Since $\ab{\varphi}\leq \advice$, the maximum number of bits read by $\ALGp$ is $\advice+O(\log n)+2=\advice+O(\log n)$.
\end{proof}

\begin{theorem}
\VC is \Wc.
\end{theorem}
\begin{proof}
By Lemma~\ref{VCupperlemma}, \VC is in \W. Combining Lemma~\ref{vckh} and Theorem~\ref{lminknown} shows that a $c$-competitive algorithm for \VC must read at least $B(n,c)-O(\log n)$ bits of advice. Thus, \VC is \Wc.
\end{proof}


\subsubsection{Online Cycle Finding.}
Most of the graph problems that we prove to be \Wc are, in their offline versions, $\mathsf{NP}$-complete. However, in this section, we show that \CF is also \Wc. The offline version of this problem is very simple and can easily be solved in polynomial time.



\begin{lemma} \label{cfub}
\CF is in \W.
\end{lemma}
This and the following proofs of membership of \W have been omitted. They are almost identical to the proof of Lemma~\ref{VCupperlemma}.

In order to show that \CF is \Wminc, we will make use of the following graph.

\begin{definition}
\label{H}
For a string $x=x_1\ldots x_n\in\{0,1\}^n$ define $f(x_i)$ to be the largest $j<i$ such that $x_j=1$.
Note that this may not always be defined. We let $\textsc{max}$ be the largest $i$ such that $x_i=1$. Similarly, we let $\textsc{min}$ be the smallest $i$ such that $x_i=1$.
We now define the graph $H_x=(V,E)$:
\begin{align*}
V&=\{v_1,\ldots , v_n\},\\
E&=\{(v_j, v_i) \colon f(x_i)=j\} \cup \{ (v_{\textsc{min}}, v_\textsc{max})\}.
\end{align*} 
Furthermore, let $V_0=\{v_i \colon x_i=0\}$ and $V_1=\{v_i \colon x_i =1\}$.
\end{definition}

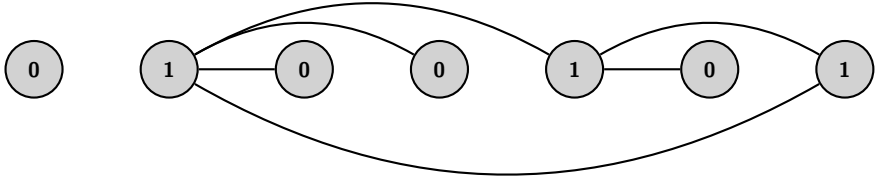
\begin{figure}[h]
\centering
\begin{tikzpicture}[-,=stealth',minimum size=0.75cm,node distance=2cm,
  thick,main node/.style={circle,fill=blue!20,minimum size=0.75cm,draw,font=\sffamily\Large\bfseries}]
\selectcolormodel{gray}
  \node[main node] (1) {{\small 0}};
  \node[main node] (2) [right=1 cm of 1] {{\small 1}};
  \node[main node] (3) [right=1 cm of 2] {{\small 0}};
  \node[main node] (4) [right=1 cm of 3] {{\small 0}};
  \node[main node] (5) [right=1 cm of 4] {{\small 1}};
  \node[main node] (6) [right=1 cm of 5] {{\small 0}};
  \node[main node] (7) [right=1 cm of 6] {{\small 1}};

  \path[every node/.style={font=\sffamily\small}]
	(2) edge (3)
	(2) edge [bend left] (4)
	(2) edge [bend left] (5)
	(2) edge [bend right] (7)
	(5) edge (6)
	(5) edge [bend left] (7)

	;
\end{tikzpicture}
\caption{$H_{0100101}$}
\end{figure}

\begin{lemma}
\label{cfkh}
If there is a $c$-competitive algorithm reading $\advice$ bits for \CF, then there is a $c$-competitive algorithm reading $\advice + O(\log n)$ bits for \minSk.
\end{lemma}
\begin{proof}
Let $\ALG$ be a $c$-competitive algorithm (with an additive constant $\alpha$) for \CF reading at most $\advice$ bits of advice. We will define an algorithm $\ALGp$ and an oracle \ORACLEp for \minSk such that $\ALGp$ is $c$-competitive (with the same additive constant) and reads at most $\advice+ O(\log n)$ bits of advice.

Let $x=x_1\ldots x_n$ be an input string to \minSk. The oracle \ORACLEp first writes one bit of advice to indicate if $\no{x}\leq 2$. If this is the case, \ORACLEp writes (in a self-delimiting way) the index of these at most two $1$s to the advice tape. This can be done using $O(\log n)$ bits and clearly allows \ALGp to be strictly $1$-competitive. In the rest of the proof, we will assume that there are at least three 1s in $x$.

Consider the input instance to \CF, $H_x=(V,E)$, defined in Definition~\ref{H}, where the vertices are requested in the order $\langle v_1, \ldots, v_n\rangle$. Note that the vertices $V_1$ form the only cycle in $H_x$. Thus, if an algorithm rejects a vertex from $V_1$, the subgraph induced by the vertices accepted by the algorithm cannot contain a cycle.

Let $\varphi$ be the advice read by $\ALG$, and let $V_{\ALG}$ be the vertices chosen by $\ALG$, when the $n$ vertices of $H_x$ are revealed. Since $\ALG$ is $c$-competitive, we know that $\ab{V_{\ALG}}\leq c\ab{V_1}+\alpha$.

We now define $\ALGp$. As usual, $y$ denotes the output computed by $\ALGp$. Since $\ALG$ is $c$-competitive, it must hold that $V_1\subseteq V_\ALG$. The oracle \ORACLEp writes the advice $\varphi$ to the advice tape. When $\ALGp$ receives request $i$ at the beginning of round $i$ in \minSk, it considers what $\ALG$ does when the vertex $v_i$ in $H_x$ is revealed: \ALGp answers $1$ if \ALG accepts $v_i$ and $0$ otherwise. Note that it is possible for $\ALGp$ to simulate $\ALG$ since, at the beginning of round $i$, $\ALGp$ knows $x_1\ldots x_{i-1}$. In particular, $\ALGp$ knows which edges were revealed to $\ALG$ along with the vertex $v_i$ in $H_x$. Note, however, that in order to simulate the edge from $v_{\textsc{min}}$ to $v_{\textsc{max}}$, $\ALG$ needs to know when $v_{\textsc{max}}$ is being revealed.
This can be achieved using $O(\log n)$ additional advice bits.

 Together with access to the advice $\varphi$ read by $\ALG$, this allows $\ALGp$ to simulate $\ALG$. Since $V_1\subseteq V_\ALG$, we get that $x\sq y$. Furthermore, since $\ab{V_\ALG}\leq c\ab{V_1}+\alpha$, we also get that $\no{y}\leq c\no{x}+\alpha$.
\end{proof}

\begin{theorem}
\CF is \Wcomplete
\end{theorem}
\begin{proof}
This follows from Lemmas~\ref{cfub} and \ref{cfkh} together with Theorem~\ref{lminknown}.
\end{proof}

\subsubsection{Online Dominating Set.}
In this section, we show that \DS is also \Wminc. We do not require that the vertices picked by the online algorithm form a dominating set at all times. We only require that the solution produced by the algorithm is a dominating set when the request sequence ends. Of course, this makes a difference only because we consider online algorithms with advice. 
For \VC, this issue did not arise, since it is not possible to end up with a vertex cover without maintaining a vertex cover at all times. Thus, in this aspect, \DS is more similar to \CF.

\begin{lemma}
\label{dsup}
\DS in in \Wmin.
\end{lemma}
In order to show that \DS is \Wminc, we use the following construction. 

\begin{definition}
\label{K}
For a string $x=x_1 \ldots x_n$ such that $\no{x}\geq 1$, define $\textsc{max}$ to be the largest $i$ such that $x_i=1$ and define $K_x=(V,E)$ as follows:
\begin{align*}
V&=\{v_1,\ldots , v_n\},\\
E&=\{(v_i, v_{\textsc{max}}) \colon x_i=0\}.
\end{align*} 
Furthermore, let $V_0=\{v_i \colon x_i=0\}$ and $V_1=\{v_i \colon x_i =1\}$.
\end{definition}

Note that $V_1$ is a smallest dominating set in $K_x$ and that any dominating set is either a superset of $V_1$ or equal to $V \setminus \{v_{\textsc{max}}\}$. 
We now give a lower bound on the advice complexity of \DS. Interestingly, it is possible to do this by making a reduction from \minSu (instead of \minSk) to \DS.

\begin{lemma} \label{dshard}
If there is a $c$-competitive algorithm for \DS reading $\advice$ bits of advice, then there is a $c$-competitive algorithm reading $\advice+ O(\log n)$ bits
of advice for \minSu.
\end{lemma}
\begin{proof}
Let $\ALG$ be a $c$-competitive algorithm (with an additive constant of $\alpha$) for \DS reading at most $\advice$ bits of advice. We will define an algorithm $\ALGp$ and an oracle \ORACLEp for \minSu such that $\ALGp$ is $c$-competitive (with the same additive constant) and reads at most $\advice+ O(\log n)$ bits of advice.

Let $x=x_1\ldots x_n$ be an input string to \minSu. The oracle \ORACLEp first writes one bit of advice to indicate if $\no{x}=0$. If this is the case, \ALGp answers 0 in each round. In the rest of the proof, we will assume that $\no{x}\geq 1$.

Consider the input instance to \DS, $K_x=(V,E)$, defined in Definition~\ref{K}, where the vertices are requested in the order $\langle v_1, \ldots, v_n \rangle$. Note that $V_1$ is the smallest dominating set in $K_x$.
Let $\varphi$ be the advice read by $\ALG$, and let $V_{\ALG}$ be the vertices chosen by $\ALG$, when the $n$ vertices of $K_x$ are revealed. Since $\ALG$ is $c$-competitive, we know that $V_{\ALG}$ is a dominating set of size $\ab{V_{\ALG}}\leq c\ab{V_1}+\alpha$. 

We now define $\ALGp$ and $\ORACLEp$. 
The second bit of the advice tape will be used to let \ALGp distinguish the two cases described below.
Note that the only vertex from $V_1$ that can be rejected by a $c$-competitive algorithm is $v_{\textsc{max}}$, and nothing can be rejected when $V_1=V$. Hence the two cases are exhaustive.

{\em Case 1: \ALG accepts all vertices in $V_1$.}
The oracle \ORACLEp writes the value of $\textsc{max}$ in a self-delimiting way. This requires $O(\log n)$ bits. Furthermore, \ORACLEp writes $\varphi$ to the advice tape.
Now, \ALGp learns $\varphi$ and $\textsc{max}$ and works as follows: In round $i\leq \textsc{max}-1$, \ALGp answers 1 if $\ALG$ accepts the vertex $v_i$ and 0 otherwise. Note that \ALGp knows that no edges are revealed to \ALG in the first $\textsc{max}-1$ rounds. Thus, \ALGp can compute the answer produced by \ALG in these rounds from $\varphi$ alone. In round $\textsc{max}$, \ALGp answers $1$. In rounds $\textsc{max}+1,\ldots , n$, the algorithm \ALGp always answers $0$. 

{\em Case 2: \ALG rejects $v_{\textsc{max}}$.}
In order to dominate $v_{\textsc{max}}$, \ALG must accept a vertex $v_i \in V_0$.
The oracle \ORACLEp writes the values of $\textsc{max}$ and $i$ in a self-delimiting way, followed by $\varphi$, to the advice tape.
\ALGp behaves as in Case 1, except that it answers 0 in round $i$.

In both cases, $x\sq y$ and $\no{y}\leq \ab{V_{\ALG}}\leq c\ab{V_{1}}+\alpha = c\no{x}+\alpha$. Furthermore, \ALGp reads $\advice + O(\log n)$ bits of advice.
\end{proof}

\begin{theorem}
\DS is \Wcomplete
\end{theorem}
\begin{proof}
This follows from Lemmas~\ref{dsup} and \ref{dshard} together with Theorem~\ref{lminknown}.
\end{proof}

\subsubsection{Online Set Cover.}
We study a version of \SC in which the universe is known from the beginning and the sets arrive online. 
Note that this problem is very different from the set cover problem studied in \cite{setcover, Asetcover}, where the elements (and not the sets) arrive online. 

\begin{lemma} \label{scw}
\SC is in \W
\end{lemma}

\begin{lemma} \label{schard}
If there is a $c$-competitive algorithm for \SC reading $\advice$ bits of advice, then there
is a $c$-competitive algorithm reading $\advice + O(\log n)$ bits of advice for \minSu.
\end{lemma}
\begin{proof}
Let $x=x_1\ldots x_n$ be an input string to \minSu with $\no{x}\geq 1$, and define {\sc max} as in Definition~\ref{K}. We define an instance of \SC as follows.
The universe is $[n]=\{1,\ldots , n\}$ and there are $n$ requests. For $i \neq \textsc{max}$, request $i$ is just the singleton $\{i\}$. Request $\textsc{max}$ is the set $\{\textsc{max}\}\cup S_0$, where $S_0=\{i : x_i=0\}$.

Using these instances of \SC and the same arguments as in Lemma~\ref{dshard} proves the theorem.
Note that for \SC, only Case 1 of Lemma~\ref{dshard} is relevant, since a $c$-competitive algorithm for this problem will accept all requests $i \not\in S_0$.
\end{proof}

\begin{theorem}
\SC is \Wc
\end{theorem}
\begin{proof}
This follows from Lemmas~\ref{scw} and \ref{schard} together with Theorem~\ref{lminknown}.
\end{proof}


\subsection{\Wmaxc maximization problems}
In this section, we consider two maximization problems which are \Wmaxc.
\subsubsection{Online Independent Set.} The first maximization problem that we consider is \IS.

\begin{lemma}
\label{isupper}
\IS is in \Wmax.
\end{lemma}
\begin{proof}
Each request is a vertex along with the edges connecting it to previously requested vertices. The algorithm makes a binary choice for each request,
to include the vertex ($y_i=0$) or not ($y_i=1$). The feasible outputs are those that are independent sets. There exists a feasible output (taking no vertices).
The score of a feasible output is the number of vertices in it, and the score of an infeasible output is $-\infty$. Any subset of the vertices in an optimal solution is a feasible solution.
\end{proof}

\begin{lemma}
\label{iskh}
If there is a $c$-competitive algorithm reading $b$ bits for \IS, then there is a $c$-competitive algorithm reading $b + O(\log n)$ bits for \maxSk.
\end{lemma}

\begin{proof}
The proof is almost identical to the proof of Lemma \ref{vckh}.
Let \ALG be a $c$-competitive algorithm (with an additive constant of $\alpha$) for \IS reading at most $b$ bits of advice. We will define an algorithm \ALGp and an oracle \ORACLEp for \maxSk such that $\ALGp$ is $c$-competitive (with the same additive constant) and reads at most $b$ bits of advice.

As in Lemma \ref{vckh}, on input $x=x_1\ldots x_n$ to \IS, the algorithm \ALGp simulates \ALG on $G_x$ (from Definition~\ref{G}). This time, a vertex in $V_0$ is \emph{good} and a vertex in $V_1$ is \emph{bad}. Note that $V_0\cup\{v_n\}$ is a maximum independent set in $G_x$. Also, if \ALG accepts a bad vertex, $v_i$, then no further vertices $v_j$ (where $i<j\leq n$) can be accepted because of the edges $(v_i, v_j)$. Thus, \ALG accepts at most one bad vertex. Let $V_{\ALG}$ be the vertices accepted by \ALG. Since \ALG is $c$-competitive, $V_{\ALG}$ is an independent set satisfying $\ab{V_0}\leq \ab{V_0\cup\{v_n\}}\leq c\ab{V_{\ALG}}+\alpha$. We denote by $y$ the output computed by $\ALGp$. There are three cases to consider:

\emph{Case 1: All vertices accepted by $\ALG$ are good, that is $V_\ALG\subseteq V_0$.} In this case, \ALG answers $0$ in round $i$ if $v_i\in V_\ALG$ and $1$ otherwise. Clearly, $x \sq y$ and $\nz{x} = \ab{V_0}\leq c\ab{V_{\ALG}}+\alpha= c\nz{y}+\alpha$.
 

\emph{Case 2a: \ALG accepts a bad vertex, $v_i$, and rejects a good vertex, $v_j$.} In this case, the oracle \ORACLEp writes the indices $i$ and $j$ in a self-delimiting way. \ALG simulates \ALGp as before, but answers $1$ in round $i$ and $0$ in round $j$. It follows that $x \sq y$ and $\nz{x}\leq c\nz{y}+\alpha$.

\emph{Case 2b: \ALG accepts a bad vertex, $v_i$, and all good vertices.} This implies that $V_\ALG$ is an independent set of size $\ab{V_0}+1$, which must be optimal. The oracle \ORACLEp writes the value of $i$ to the advice tape in a self-delimiting way. \ALG simulates \ALGp as before but answers $1$ in round $i$. It follows that $x \sq y$. Furthermore, $\nz{y}=\ab{V_{\ALG}}-1=\ab{V_0}=\nz{x}$, and hence the solution $y$ is optimal.

In order to simulate \ALG, the algorithm \ALGp needs to read at most $b$ bits of advice plus $O(\log n)$ bits of advice to specify the case and handle the cases where \ALG accepts a bad vertex. 
\end{proof}

\begin{theorem}
\IS is \Wc.
\end{theorem}
\begin{proof}
This follows from Lemmas~\ref{isupper} and \ref{iskh} together with Theorem~\ref{kpacklow}.
\end{proof}

\subsubsection{Online Disjoint Path Allocation.}
In this section, we show that \DPA is \Wmaxc.

\begin{lemma}
\DPA is in \Wmax.
\end{lemma}

In Lemma~\ref{dpkh}, we use the same hard instance for \DPA as in \cite{A1} to get a lower bound on the advice complexity of \DPA.
\begin{lemma}
\label{dpkh}
If there is a $c$-competitive algorithm reading $b$ bits for \DPA, then there is a $c$-competitive algorithm reading $b + O(\log n)$ bits for \maxSk.
\end{lemma}
\begin{proof}
The proof is similar to the proof of Lemma \ref{iskh}. Let \ALG be a $c$-competitive algorithm for \DPA reading at most $b$ bits of advice. We will describe an algorithm \ALGp and an oracle \ORACLEp for \maxSk such that $\ALGp$ is $c$-competitive (with the same additive constant $\alpha$ as $\ALG$) and reads at most $b$ bits of advice. 

Let $x=x_1\ldots x_n$ be an input to \maxSk. We define an instance $I_x$ of \DPA with $L=2^{n}$ (that is, the number of vertices on the path is $2^{n}+1$).
In round $i$, $1 \leq i \leq n$, a path of length $2^{n-i}$ arrives.
For $2 \leq i \leq n$, the position of the path depends on $x_{i-1}$.
We define the request sequence inductively
(for an example, see Figure~\ref{fig:odpa}): 

\begin{figure}
\begin{align*}
&x=010\\
&L=8\\
&I_x = \langle (0,4), (4,6), (4,5) \rangle
\end{align*}
\caption{\label{fig:odpa}An example of the reduction used in the proof of Lemma~\ref{dpkh}. The request $(0,4)$ is good, since $x_1=0$, and $(4,6)$ is a bad request, since $x_2=1$.}
\end{figure}

In round $1$, the request $(u_1,v_1)$ arrives, where 
\begin{align*}
&u_1=0\\
&v_1=2^{n-1}
\end{align*}

In round $i$, $2\leq i\leq n$, the request $(u_i, v_i)$ arrives, where
\begin{align*}
&u_i=
  \begin{cases}
    u_{i-1}, & \text{if $x_{i-1}=1$}\\ 
    v_{i-1}, & \text{if $x_{i-1}=0$} 
  \end{cases}\\
&v_i=u_i+2^{n-i}
\end{align*}

We say that a request $r_i=(u_i, v_i)$ is \emph{good} if $x_i=0$ and \emph{bad} if $x_i=1$. If $r_i$ is good, then none of the later requests overlap with $r_i$. On the other hand, if $r_i$ is bad, then all later requests do overlap with $r_i$. In particular, if one accepts a single bad request, then no further requests can be accepted. An optimal solution is obtained if one accepts all good requests together with $r_n$.

The oracle \ORACLEp will provide \ALGp with the advice $\varphi$ read by \ALG when processing $I_x$. Since \ALG knows the value of $x_{i-1}$ at the beginning of round $i$ in \maxSk, \ALG can use the advice $\varphi$ to simulate \ALGp on $I_x$. If \ALG only accepts good requests, it is clear that \ALGp can compute an output $y$ such that $x \sq y$ and $\nz{x}\leq c\nz{y}+\alpha$. The case where \ALG accepts a bad request is handled by using at most $O(\log n)$ additional advice bits, exactly as in Lemma \ref{iskh}.
\end{proof}

\begin{theorem}
\DPA is \Wmaxc.
\end{theorem}


\subsection{\W Problems which are not \Wc}
In this section, we will give two examples of problems in \W which are provably not \Wc.
\subsubsection{Uniform knapsack.} 

We define the problem \KS as follows: For each request, $i$, an item of weight $a_i$, $0\leq a_i\leq 1$, is requested. A request must immediately be either accepted or rejected, and this decision is irrevocable. Let $S$ denote the set of indices of accepted items. We say that $S$ is a feasible solution if $\sum_{i\in S}a_i\leq 1$. The profit of a feasible solution is the number of items accepted (all items have a value of $1$). The problem is a maximization problem. 

The \KS problem is the online knapsack problem as studied in \cite{Aknapsack}, but with the restriction that all items have a value of $1$. This problem is the same as online dual bin packing with only a single bin available and where items can be rejected. 

It is clear that \KS belongs to \Wmax since a subset of a feasible solution is also a feasible solution. Furthermore, since all items have value $1$, the profit of a feasible solution is simply the number of items packed in the knapsack. The problem is hard in the sense that no deterministic algorithm (without advice) can attain a strict competitive ratio better than $\Omega (n)$ (see \cite{MVknapsack, Aknapsack}). However, as the next lemma shows, the problem is not \Wmaxc. In \cite{Aknapsack}, it is shown that for any $\varepsilon>0$, it is possible to achieve a competitive ratio of $1+\varepsilon$ using $O(\log n)$ bits of advice, under the assumption that all weights and values can be represented in polynomial space. Lemma~\ref{uks} shows how this assumption can be avoided when all items have unit value.
\begin{lemma}
\label{uks}
There is a strictly $2$-competitive \KS algorithm reading $O(\log n)$ bits of advice, where $n$ is the length of the input.
\end{lemma}
\begin{proof}
Fix an input $\sigma= \langle a_1,\ldots , a_n \rangle$. Let $m$ be the number of items accepted by \OPT. The oracle writes $m$ to the advice tape using a self-delimiting encoding. Since $m\leq n$, this requires $O(\log n)$ bits. The algorithm \ALG learns $m$ from the advice tape and works as follows: If $\ALG$ is offered an item, $a_i$, such that $a_i\leq 2/m$ and if accepting $a_i$ will not make the total weight of $\ALG$'s solution larger than $1$, then \ALG accepts $a_i$. Otherwise, $a_i$ is rejected.

In order to show that $\ALG$ is strictly $2$-competitive, we define $A=\{a_i : a_i\leq 2/m\}$. First note that $\ab{A} \geq m/2$, since the sizes of the $m$ smallest items add up to at most 1. Thus, if \ALG accepts all items contained in $A$, it accepts at least $m/2$ items.
On the other hand, if \ALG rejects any item $a_i \in A$, it means that it has already accepted items of total size more than $1 - 2/m$. Since all accepted items have size at most $2/m$, this means that \ALG has accepted at least $m/2$ items.
\end{proof}

Even though \KS is not \Wmaxc, the fact that it belongs to \Wmax might still be of interest, since this provides some starting point for determining the advice complexity of the problem. In particular, it gives some (non-trivial) way to obtain a $c$-competitive algorithm for $c<2$. Determining the exact advice complexity of \KS is left as an open problem.

\subsubsection{Matching under edge-arrival.}
We briefly consider the \OM problem in an edge-arrival version. For each request, an edge is revealed. An edge can be either accepted or rejected. Denote by $E_{\ALG}$ the edges accepted by some algorithm \ALG. A solution $E_{\ALG}$ is feasible if the set of edges in the solution is a matching in the input graph, and the profit of a feasible solution is the number of edges in $E_{\ALG}$. The problem is a maximization problem.

It is well-known that the greedy algorithm is $2$-competitive for \OM. Since this algorithm works in an online setting without any advice, it follows that \OM is not \Wc. On the other hand, \OM is in \W. This gives an upper bound on the advice complexity of the problem for $1\leq c<2$. It seems obvious that this upper bound is not tight, but currently, no better bound is known.
\section{Conclusion and Open Problems}
The following theorem summarizes the main results of this paper.
\begin{theorem}
For the problems
\begin{itemize}
\item \VC
\item \CF
\item \DS
\item \SC (set-arrival version)
\item \IS
\item \DPA
\end{itemize}
and for any $c>1$, possibly a function of the input length $n$,
\begin{align*}
	b=\log\left(1+\frac{(c-1)^{c-1}}{c^{c}}\right)n \pm O(\log n)
\end{align*}
bits of advice are necessary and sufficient to achieve a (strict) competitive ratio of $c$.
\end{theorem}

As with the original string guessing problem \oldS~\cite{A2, SG}, we have shown that \S is a useful tool for determining the advice complexity of online problems. It seems plausible that one could identify other variants of online string guessing and obtain classes similar to \W. Potentially, this could lead to an entire hierarchy of string guessing problems and related classes. 

More concretely, there are various possibilities of generalizing \S. One could associate some positive weight to each bit $x_i$ in the input string. The goal would then be to produce a feasible output of minimum (or maximum) weight. Such a string guessing problem would model minimum weight vertex cover (or maximum weight independent set). Note that for \maxS, the algorithm from Theorem \ref{trivialmax} works in the weighted version. However, the same is not true for any of the algorithms we have given for \minS. Thus, it remains an open problem if $O(n/c)$ bits of advice suffice to achieve a competitive ratio of $c$ for the weighted version of \minS. 

\begin{acknowledgements}
The authors would like to thank Magnus Gausdal Find for helpful discussions.
\end{acknowledgements}



\bibliographystyle{spmpsci}
\bibliography{refs}

\begin{thebibliography}{10}
\providecommand{\url}[1]{{#1}}
\providecommand{\urlprefix}{URL }
\expandafter\ifx\csname urlstyle\endcsname\relax
  \providecommand{\doi}[1]{DOI~\discretionary{}{}{}#1}\else
  \providecommand{\doi}{DOI~\discretionary{}{}{}\begingroup
  \urlstyle{rm}\Url}\fi

\bibitem{setcover}
Alon, N., Awerbuch, B., Azar, Y., Buchbinder, N., Naor, J.: The online set
  cover problem.
\newblock SIAM J. Comput. \textbf{39}(2), 361--370 (2009)

\bibitem{Asteiner}
Barhum, K.: Tight bounds for the advice complexity of the online minimum
  steiner tree problem.
\newblock In: Proc. 40th International Conf. on Current Trends in Theory and
  Practice of Computer Science (SOFSEM), \emph{Lecture Notes in Comput. Sci.,
  Springer}, vol. 8327, pp. 77--88 (2014)

\bibitem{DPAL}
Barhum, K., B{\"o}ckenhauer, H.J., Fori{\v{s}}ek, M., Gebauer, H.,
  Hromkovi{\v{c}}, J., Krug, S., Smula, J., Steffen, B.: On the power of advice
  and randomization for the disjoint path allocation problem.
\newblock In: Proc. 40th International Conf. on Current Trends in Theory and
  Practice of Computer Science (SOFSEM), \emph{Lecture Notes in Comput. Sci.,
  Springer}, vol. 8327, pp. 89--101 (2014)

\bibitem{Avbipartite}
Bianchi, M.P., B{\"o}ckenhauer, H.J., Hromkovi{\v{c}}, J., Keller, L.: Online
  coloring of bipartite graphs with and without advice.
\newblock Algorithmica \textbf{70}(1), 92--111 (2014)

\bibitem{SG}
B{\"o}ckenhauer, H.J., Hromkovi{\v{c}}, J., Komm, D., Krug, S., Smula, J.,
  Sprock, A.: The string guessing problem as a method to prove lower bounds on
  the advice complexity.
\newblock Theor. Comput. Sci. \textbf{554}, 95--108 (2014)

\bibitem{Ak-server}
B{\"o}ckenhauer, H.J., Komm, D., Kr\'alovi\v{c}, R., Kr\'alovi\v{c}, R.: On the
  advice complexity of the k-server problem.
\newblock In: Proc. 38th International Colloquium on Automata, Languages, and
  Programming (ICALP), \emph{Lecture Notes in Comput. Sci., Springer}, vol.
  6755, pp. 207--218 (2011)

\bibitem{A1}
B{\"o}ckenhauer, H.J., Komm, D., Kr\'alovi\v{c}, R., Kr\'alovi\v{c}, R.,
  M{\"o}mke, T.: On the advice complexity of online problems.
\newblock In: Proc. 20th International Symp. on Algorithms and Computation
  (ISAAC), \emph{Lecture Notes in Comput. Sci., Springer}, vol. 5878, pp.
  331--340 (2009)

\bibitem{Aknapsack}
B{\"o}ckenhauer, H.J., Komm, D., Kr\'alovi\v{c}, R., Rossmanith, P.: The online
  knapsack problem: Advice and randomization.
\newblock Theor. Comput. Sci. \textbf{527}, 61--72 (2014)

\bibitem{AListupdate}
Boyar, J., Kamali, S., Larsen, K.S., L{\'o}pez-Ortiz, A.: On the list update
  problem with advice.
\newblock In: Proc. 8th International Conf. on Language and Automata Theory and
  Applications (LATA), \emph{Lecture Notes in Comput. Sci., Springer}, vol.
  8370, pp. 210--221 (2014).
\newblock Full paper to appear in {\em Information and Computation}.

\bibitem{Abpj}
Boyar, J., Kamali, S., Larsen, K.S., L{\'o}pez-Ortiz, A.: Online bin packing
  with advice.
\newblock Algorithmica \textbf{74}, 507--527 (2016)

\bibitem{Demange}
Demange, M., Paschos, V.T.: On-line vertex-covering.
\newblock Theor. Comput. Sci. \textbf{332}(1-3), 83--108 (2005)

\bibitem{opac-b1088981}
Dinitz, J.H., Stinson, D.R. (eds.): Contemporary {D}esign {T}heory: a
  {C}ollection of {S}urveys.
\newblock Wiley-Interscience series in discrete mathematics and optimization.
  Wiley, New York (1992).
\newblock \urlprefix\url{http://opac.inria.fr/record=b1088981}

\bibitem{Ais}
Dobrev, S., Kr\'alovi\v{c}, R., Kr\'alovi\v{c}, R.: Advice complexity of
  maximum independent set in sparse and bipartite graphs.
\newblock Theory Comput. Syst. \textbf{56}(1), 197--219 (2015)

\bibitem{A4}
Dobrev, S., Kr\'alovi\v{c}, R., Pardubsk{\'a}, D.: Measuring the
  problem-relevant information in input.
\newblock RAIRO - Theor.\ Inf.\ Appl. \textbf{43}(3), 585--613 (2009)

\bibitem{A2}
Emek, Y., Fraigniaud, P., Korman, A., Ros{\'e}n, A.: Online computation with
  advice.
\newblock Theor. Comput. Sci. \textbf{412}(24), 2642--2656 (2011)

\bibitem{erdos-spencer}
Erd{\H o}s, P., Spencer, J.: Probabilistic Methods in Combinatorics.
\newblock Academic Press (1974)

\bibitem{Avpath}
Fori{\v{s}}ek, M., Keller, L., Steinov{\'a}, M.: Advice complexity of online
  coloring for paths.
\newblock In: Proc. 6th International Conf. on Language and Automata Theory and
  Applications (LATA), \emph{Lecture Notes in Comput. Sci., Springer}, vol.
  7183, pp. 228--239 (2012)

\bibitem{Asushmita}
Gupta, S., Kamali, S., L{\'o}pez-Ortiz, A.: On advice complexity of the
  k-server problem under sparse metrics.
\newblock In: Proc. 20th International Colloquium on Structural Information and
  Communication Complexity (SIROCCO), \emph{Lecture Notes in Comput. Sci.,
  Springer}, vol. 8179, pp. 55--67 (2013)

\bibitem{Magnus}
Halld{\'o}rsson, M.M., Iwama, K., Miyazaki, S., Taketomi, S.: Online
  independent sets.
\newblock Theor. Comput. Sci. \textbf{289}(2), 953--962 (2002)

\bibitem{MagnusSzegedy}
Halld{\'o}rsson, M.M., Szegedy, M.: Lower bounds for on-line graph coloring.
\newblock Theor. Comput. Sci. \textbf{130}(1), 163--174 (1994)

\bibitem{Haastad}
H{\aa}stad, J.: Clique is hard to approximate within {$n^{1-\epsilon}$}.
\newblock Acta Math. \textbf{182}(1), 105--142 (1999)

\bibitem{A3}
Hromkovi{\v{c}}, J., Kr\'alovi\v{c}, R., Kr\'alovi\v{c}, R.: Information
  complexity of online problems.
\newblock In: Proc. 35th Symp. on Mathematical Foundations of Computer Science
  (MFCS), \emph{Lecture Notes in Comput. Sci., Springer}, vol. 6281, pp. 24--36
  (2010)

\bibitem{CompRatio1}
Karlin, A.R., Manasse, M.S., Rudolph, L., Sleator, D.D.: Competitive snoopy
  caching.
\newblock Algorithmica \textbf{3}, 77--119 (1988)

\bibitem{Asetcover}
Komm, D., Kr\'alovi\v{c}, R., M{\"o}mke, T.: On the advice complexity of the
  set cover problem.
\newblock In: Proc. 7th International Computer Science Symp. in Russia (CSR),
  \emph{Lecture Notes in Comput. Sci., Springer}, vol. 7353, pp. 241--252
  (2012)

\bibitem{MVknapsack}
Marchetti-Spaccamela, A., Vercellis, C.: Stochastic on-line knapsack problems.
\newblock Math. Program. \textbf{68}, 73--104 (1995)

\bibitem{Aec}
Mikkelsen, J.W.: Optimal online edge coloring of planar graphs with advice.
\newblock In: Proc. 9th International Conf. on Algorithms and Complexity
  (CIAC), \emph{Lecture Notes in Comput. Sci., Springer}, vol. 9079, pp.
  352--364 (2015)

\bibitem{Mitzen}
Mitzenmacher, M., Upfal, E.: Probability and Computing - Randomized Algorithms
  and Probabilistic Analysis.
\newblock Cambridge University Press (2005)

\bibitem{Abip}
Miyazaki, S.: On the advice complexity of online bipartite matching and online
  stable marriage.
\newblock Inf. Process. Lett. \textbf{114}(12), 714--717 (2014)

\bibitem{Raz}
Raz, R., Safra, S.: A sub-constant error-probability low-degree test, and a
  sub-constant error-probability {PCP} characterization of {NP}.
\newblock In: Proc. 29th Symp. on Theory of Computing (STOC), pp. 475--484. ACM
  (1997)

\bibitem{Aschedule}
Renault, M.P., Ros{\'e}n, A., van Stee, R.: Online algorithms with advice for
  bin packing and scheduling problems.
\newblock Theor. Comput. Sci. \textbf{600}, 155--170 (2015)

\bibitem{Avtripartite}
Seibert, S., Sprock, A., Unger, W.: Advice complexity of the online coloring
  problem.
\newblock In: Proc. 8th International Conf. on Algorithms and Complexity
  (CIAC), \emph{Lecture Notes in Comput. Sci., Springer}, vol. 7878, pp.
  345--357 (2013)

\bibitem{CompRatio2}
Sleator, D.D., Tarjan, R.E.: Amortized efficiency of list update and paging
  rules.
\newblock Commun. ACM \textbf{28}(2), 202--208 (1985)

\end{thebibliography}

\clearpage
\appendix
\normalsize
\section*{Appendix}
\section{Approximation of the Advice Complexity Bounds}

\label{calcbounds}

In Theorems \ref{covup}-\ref{packlow}, bounds on the advice complexity of \S were obtained. These bounds are tight up to an additive term of $O(\log n)$. However, within the proofs, they are all expressed in terms of the minimum size of a certain covering design or a quotient of binomial coefficients. In this appendix, we prove the closed formula estimates for the advice complexity stated in Theorems \ref{covup}-\ref{packlow} and \ref{algsgeasy}. Again, these estimates are tight up to an additive term of $O(\log n)$. The key to obtaining the estimates is the estimation of a binomial coefficient using the binary entropy function.

\subsection{Approximating the Function $B(n,c)$}

\begin{lemma}
For $c>1$, it holds that $$\frac{1}{e\ln(2)}\frac{1}{c}\leq\log\left(1+\frac{(c-1)^{c-1}}{c^{c}}\right)\leq \frac{1}{c}.$$
\label{simplelemma}
\end{lemma}
\begin{proof}
We prove the upper bound first.
To this end, note that
\begin{align*}
&\log\left(1+\frac{(c-1)^{c-1}}{c^{c}}\right)\leq \frac{1}{c} \;
\Leftrightarrow \; 1+\frac{(c-1)^{c-1}}{c^{c}}\leq 2^{1/c} \;
\Leftrightarrow \; \left(1+\frac{(c-1)^{c-1}}{c^{c}}\right)^c\leq 2.
\end{align*}
Using calculus, one may verify that $\left(1+\frac{(c-1)^{c-1}}{c^{c}}\right)^c$ is decreasing in $c$ for $c>1$. Thus, by continuity, it follows that
\begin{align*}
\left(1+\frac{(c-1)^{c-1}}{c^{c}}\right)^c&\leq\lim_{c\rightarrow 1^+}\left(1+\frac{(c-1)^{c-1}}{c^{c}}\right)^c
=\lim_{c\rightarrow 1^+}\left(1+\left(\frac{c-1}{c}\right)^{c-1} \, \frac{1}{c}\right)^c\\
&=\lim_{c\rightarrow 1^+}\left(1+\frac{1}{c}\right)^c
 =2.
\end{align*}

For the lower bound, let $a=e\ln (2)$ and note that
\begin{align*}
&\frac{1}{ac}\leq \log\left(1+\frac{(c-1)^{c-1}}{c^{c}}\right) \;
\Leftrightarrow \; 2\leq \left(1+\frac{(c-1)^{c-1}}{c^{c}}\right)^{ac}.
\end{align*}
Again, using calculus, one may verify that $\left(1+\frac{(c-1)^{c-1}}{c^{c}}\right)^{ac}$ is decreasing in $c$ for $c>1$. It follows that
\begin{align*}
\left(1+\frac{(c-1)^{c-1}}{c^{c}}\right)^{ac}&\geq \lim_{c\rightarrow\infty}\left(1+\frac{(c-1)^{c-1}}{c^{c}}\right)^{ac}
=\lim_{c\rightarrow\infty}\left(1+\left(\frac{c-1}{c}\right)^{c-1} \, \frac{1}{c}\right)^{ac}\\
&=\lim_{c\rightarrow\infty}\left(1+\frac{1}{e} \, \frac{1}{c}\right)^{ac}
=\lim_{c\rightarrow\infty}\left(1+\frac{a/e}{ac}\right)^{ac}
= e^{a/e}=e^{\ln (2)}=2.
\end{align*}
\end{proof}

\subsection{The Binary Entropy Function}

In this section, we give some properties of the binary entropy function that will be used extensively in Section~\ref{approxmin}.

\begin{definition}
\label{entropy}
The \emph{binary entropy function} $H:[0,1]\rightarrow [0,1]$ is the function given by
$$H(p)=-p\log(p) - (1-p)\log(1-p),\;\text{for }0<p<1,$$ 
and $H(0)=H(1)=0$.
\end{definition}

\begin{lemma}[Lemma 9.2 in \cite{Mitzen}]
\label{entro}
For integers $m,n$ such that $0\leq m \leq n$,
\begin{equation*}
\frac{2^{nH\left(m/n\right)}}{n+1}\leq \binom{n}{m}\leq 2^{nH\left(m/n\right)}.
\end{equation*}
\end{lemma}


\begin{proposition}
The binary entropy function $H(p)$ has the following properties.
\begin{enumerate}[(H1)]
\item $H\left(\frac{1}{s}\right)=\log (s)+\frac{1-s}{s}\log(s-1)$ for $s>1$.
\item $s H\left(\frac{1}{s}\right)\leq \log s + 2$ for $s>1$.
\item $H'(p)=\log\left(\frac{1}{p}-1\right)$ and $H''(p)<0$ for $0<p<1$.
\item For any fixed $t>0$, $sH\left(\frac{t}{s}\right)$ is increasing in $s$ for $s>t$.
\item $nH\left(\frac{1}{x}\right)-nH\left(\frac{1}{x}+\frac{1}{n}\right)< 3$ if $n\geq 3$ and $x>2$.
\end{enumerate}
\end{proposition}
\begin{proof}
\emph{(H1):} Follows from the definition.

\emph{(H2):} For $s>1$,
\begin{align*}
sH\left(\frac{1}{s}\right)&=s\left(\log s + \frac{1-s}{s}\log (s-1)\right), \text{ by {\em (H1)}}\\
&=\log\left(\left(1+\frac{1}{s-1}\right)^{s-1}s\right)
\leq \log (e\cdot s)= \log (e) + \log (s)\leq \log s +2.
\end{align*}

\emph{(H3):} Note that $H$ is smooth for $0<p<1$. The derivative $H'(p)$ can be calculated from the definition. The second-order derivative is $$H''(p)=\frac{-1}{(1-p)p\ln (2)}\,,$$ which is strictly less than zero for all $0<p<1$.

\emph{(H4):} Fix $t>0$. The claim follows by showing that the partial derivative of $sH(\frac{t}{s})$ with respect to $s$ is positive for all $s>t$.

\begin{align*}
\frac{d}{ds}\left(sH\left(\frac{t}{s}\right)\right)&=H\left(\frac{t}{s}\right)+sH'\left(\frac{t}{s}\right)\left(-\frac{t}{s^2}\right)
=H\left(\frac{t}{s}\right)-\frac{t}{s}H'\left(\frac{t}{s}\right)\\
&=-\frac{t}{s}\log\left(\frac{t}{s}\right)-\left(1-\frac{t}{s}\right)\log\left(1-\frac{t}{s}\right)-\frac{t}{s}\log\left(\frac{s}{t}-1\right), \text{ by Def.~\ref{entropy} and {\em (H3)}}\\
&=-\log \left(1-\frac{t}{s}\right)>0.
\end{align*}

\emph{(H5):} $H(p)$ is increasing for $0\leq p\leq \frac{1}{2}$ and decreasing for $\frac{1}{2}\leq p\leq 1$. If $\frac{1}{x}+\frac{1}{n}\leq \frac12$, then the claim is trivially true (since then the difference is negative). Assume therefore that $\frac{1}{x}+\frac{1}{n}> \frac12$. Under this assumption, $H(\frac{1}{x})$ increases and $H(\frac{1}{x}+\frac{1}{n})$ decreases as $x$ tends to $2$. Thus, $H(\frac{1}{x})-H(\frac{1}{x}+\frac{1}{n})$ increases as $x$ tends to $2$ and, hence,

\begin{equation}
\label{Hligning1}
H\left(\frac{1}{x}\right)-H\left(\frac{1}{x}+\frac{1}{n}\right)\leq H\left(\frac{1}{2}\right)-H\left(\frac{1}{2}+\frac{1}{n}\right).
\end{equation}
Inserting into the definition of $H$ gives
\begin{align*}
H\left(\frac{1}{2}\right)-H\left(\frac{1}{2}+\frac{1}{n}\right)&=1-\left(-\left(\frac{1}{2}+\frac{1}{n}\right)\log \left(\frac{1}{2}+\frac{1}{n}\right)-\left(\frac{1}{2}-\frac{1}{n}\right)\log \left(\frac{1}{2}-\frac{1}{n}\right) \right)\\
&=\frac{1}{n} \log\left(\frac{\frac{1}{2}+\frac{1}{n}}{\frac{1}{2}-\frac{1}{n}}\right)+\frac{1}{2}\log\left(\left(\frac{1}{2}+\frac{1}{n}\right)\left(\frac{1}{2}-\frac{1}{n}\right)\right)+1\\
&=\frac{1}{n} \log\left(\frac{n+2}{n-2}\right)+\frac{1}{2}\log\left(\frac{n^2-4}{4n^2}\right)+1
\end{align*}
Since $(n+2)/(n-2)$ is decreasing for $n\geq 3$, it follows that $\log((n+2)/(n-2))\leq \log (5)$. Furthermore, $(n^2-4)/(4n^2)\leq\frac{1}{4}$ for all $n\geq 3$, and so $\frac{1}{2}\log\left((n^2-4)/(4n^2)\right)+1\leq 0$. We conclude that, for all $n\geq 3$,

\begin{equation}
H\left(\frac{1}{2}\right)-H\left(\frac{1}{2}+\frac{1}{n}\right)\leq \frac{\log(5)}{n}<\frac{3}{n}.
\label{Hligning2}
\end{equation}
Combining (\ref{Hligning1}) and (\ref{Hligning2}) proves (H5).
\end{proof}

\subsection{Binomial Coefficients}
The following proposition is a collection of simple facts about the binomial coefficient that will be used in Sections~\ref{approxmin} and \ref{approxmax}. 
\begin{proposition}
Let $a,b,c\in \mathbb{N}$.
\begin{enumerate}[(B1)]
\item \label{binomminus}$\binom{a}{b}=\frac{a}{a-b}\binom{a-1}{b}$, where $b<a$.
\item \label{binomincr}For fixed $b$, $\binom{a}{b}$ is increasing in $a$.
\item \label{binomfrac}If $c\leq b\leq a$, then $$\frac{\binom{a}{c}}{\binom{b}{c}}=\frac{\binom{a}{b}}{\binom{a-c}{a-b}}.$$
\end{enumerate}
\end{proposition}
\begin{proof}
First, we prove \Bminus:
$${a \choose b} = \frac{a!}{b!(a-b)!} = \frac{a}{a-b} \, \frac{(a-1)!}{b!(a-1-b)!} = \frac{a}{a-b} {a-1 \choose b}$$

\Bincr follows directly from \Bminus.

To prove \Bfrac, we calculate the two fractions separately:
\begin{align*}
\frac{{a \choose c}}{{b \choose c}} & = \frac{a!}{c!(a-c)!} \frac{c!(b-c)!}{b!} = \frac{a!}{(a-c)!} \frac{(b-c)!}{b!}\\
\frac{{a \choose b}}{{a-c \choose a-b}} & = \frac{a!}{b!(a-b)!} \frac{(a-b)!(b-c)!}{(a-c)!} = \frac{a!}{b!} \frac{(b-c)!}{(a-c)!} = \frac{{a \choose c}}{{b \choose c}}
\end{align*}
\end{proof}

\subsection{Approximating the Advice Complexity Bounds for \minS}
\label{approxmin}

The following lemma is used for proving Theorems~\ref{covup}--\ref{lminknown}.

\begin{lemma}
\label{calclower}
For $c>1$ and $n\geq 3$,
\begin{align}
  \log\left(\max_{t\colon \floor{ct}< n} C(n,\floor{ct},t)\right)&\geq \log\left(\max_{t\colon \floor{ct}< n}\frac{\binom{n}{t}}{\binom{\floor{ct}}{t}}\right)\label{LowerIn1}\\
&\geq \log\left(1+\frac{(c-1)^{c-1}}{c^{c}}\right)n-2\log(n+1)-5\label{LowerIn2}
\end{align}
and
\begin{align}
\log\left(\max_{t\colon \floor{ct}< n} C(n,\floor{ct},t)\right)&\leq \log\left(\max_{t\colon \floor{ct} < n}\frac{\binom{n}{t}}{\binom{\floor{ct}}{t}}n\right)\label{UpperIn1}\\
&\leq \log\left(1+\frac{(c-1)^{c-1}}{c^{c}}\right)n+3\log (n+1)\label{UpperIn2}.
\end{align}
\end{lemma}

\begin{proof}
We prove the upper and lower bounds separately.
\paragraph{Upper bound:} Fix $n,c$. By Lemma~\ref{erdos},
\begin{align*}
C(n,\floor{ct},t)&\leq\frac{\binom{n}{t}}{\binom{\floor{ct}}{t}}\left(1+\ln\binom{\floor{ct}}{t}\right).
\end{align*}
Note that $1+\ln \binom{\floor{ct}}{t}\leq n$ since we consider only $\floor{ct}<n$. This proves (\ref{UpperIn1}). Now, taking the logarithm on both sides gives
\begin{align}
\notag \log (C(n,\floor{ct},t)) &\leq \log\left(\frac{\binom{n}{t}}{\binom{\floor{ct}}{t}} \right)+\log n
\leq \log\left(\frac{\binom{n}{t}}{\binom{\ceil{ct}-1}{t}} \right)+\log n\\
\notag &\leq\; \log\left(\frac{\binom{n}{t}}{\frac{\ceil{ct}-t}{\ceil{ct}}\binom{\ceil{ct}}{t}} \right)+\log n, \text{ by \Bminus}\\
\notag &\leq \log\left(\frac{\binom{n}{t}}{\binom{\ceil{ct}}{t}} \right) + \log\left(\frac{\ceil{ct}}{\ceil{ct}-t}\right)+\log n\\
\label{Ulogc} &\leq \log\left(\frac{\binom{n}{t}}{\binom{\ceil{ct}}{t}} \right)+2\log n\,.
\end{align}
Above, we have increased $\floor{ct}$ to $\ceil{ct}$ in the binomial coefficient (at the price of an additive term of $\log n$). This is done since it will later be convenient to use that $ct\leq\ceil{ct}$. Using Lemma~\ref{entro}, we get that
\begin{align*}
 \frac{\binom{n}{t}}{\binom{\ceil{ct}}{t}} \leq \frac{2^{nH(t/n)}}{2^{\ceil{ct}H(t/\ceil{ct})}}\left(\ceil{ct}+1\right),
\end{align*}
and therefore
\begin{align}
\notag \log\left(\frac{\binom{n}{t}}{\binom{\ceil{ct}}{t}} \right)&\leq nH\left(\frac{t}{n}\right)-\ceil{ct}H\left(\frac{t}{\ceil{ct}}\right)+\log\left(\ceil{ct}+1\right)\\
\label{Ulogbin}&\leq nH\left(\frac{t}{n}\right)-ctH\left(\frac{1}{c}\right)+\log(n+1), \text{ by {\em (H4)}.}
\end{align}
Define $$M(n,t)=nH\left(\frac{t}{n}\right)-ctH\left(\frac{1}{c}\right).$$ Combining (\ref{Ulogc}) and (\ref{Ulogbin}) shows that 
\begin{equation}
\label{mpluslog}
\log (C(n, \floor{ct}, t))\leq M(n,t)+3\log (n+1)\,.
\end{equation} The function $M$ is smooth. For any given input length $n$, we can determine the value of $t$ maximizing $M(n,t)$ using calculus. In order to simplify the notation for these calculations, define
\begin{equation*}
x=\left(\frac{c}{c-1}\right)^{c}(c-1)+1,
\label{x}
\end{equation*}
and note that
\begin{align}
\notag\log(x-1)&=c\left(\log c +\frac{1-c}{c}\log (c-1)\right)\\
         &=cH\left(\frac{1}{c}\right), \text{ by {\em(H1)}.}
\label{factx}
\end{align}
We want to determine those values of $t$ for which $\frac{d}{dt}M(n,t)=0$:
\begin{align*}
&\frac{d}{dt}M(n,t)=\frac{d}{dt}\left(nH\left(\frac{t}{n}\right)-ctH\left(\frac{1}{c}\right)\right)=0\\
  \Leftrightarrow&\; nH'\left(\frac{t}{n}\right) \cdot \frac{1}{n} -cH\left(\frac{1}{c}\right)=0\\
\Leftrightarrow&\; \log\left(\frac{n}{t}-1\right)=cH\left(\frac{1}{c}\right), \text{ by {\em (H3)}}\\
\Leftrightarrow&\; \frac{n}{t}=2^{cH(1/c)}+1\\
\Leftrightarrow&\; t=\frac{n}{2^{cH(1/c)}+1}\\
\Leftrightarrow&\; t=\frac{n}{2^{\log(x-1)}+1}, \text{ by (\ref{factx})}\\
\Leftrightarrow&\; t=\frac{n}{x}.
\end{align*}
Note that \mbox{$\frac{d^2}{dt^2}M(n,t)=H''(\frac{t}{n})/n<0$} for all values of $t$, by {\em (H3)}. Thus, 
\begin{equation}
\label{mupper}
M(n,t)\leq M\left(n,\frac{n}{x}\right), \text{ for all values of } t\,.
\end{equation}
The value of $M(n, \frac{n}{x})$ can be calculated as follows: 
\begin{align}
\label{mcalc}
\notag M\left(n, \frac{n}{x}\right)&= n\,H\left(\frac{1}{x} \right)-c\,\frac{n}{x}\, H\left(\frac{1}{c}\right)\\
\notag &=n\left(\log(x)+\frac{1-x}{x}\log(x-1)-\frac{c}{x}H(1/c)\right), \text{ by {\em (H1)}}\\
\notag &=\; n\left(\log(x)+\frac{1-x}{x}\log(x-1)-\frac{1}{x}\log(x-1) \right), \text{ by (\ref{factx})}\\
\notag &=n\big(\log(x)-\log(x-1)\big)
=n\log\left(\frac{x}{x-1}\right)\\
&=n\log\left(1+\frac{(c-1)^{c-1}}{c^{c}}\right).
\end{align}
Combining (\ref{mpluslog}), (\ref{mupper}), and (\ref{mcalc}), we conclude that
\begin{displaymath}
\log(C(n, \floor{ct}, t))\leq n\log\left(1+\frac{(c-1)^{c-1}}{c^{c}}\right)+3\log (n+1).
\end{displaymath}

\paragraph{Lower Bound:}
 By Lemma~\ref{erdos}, $$\log\left(\max_{t\colon \floor{ct}< n} C(n,\floor{ct},t)\right) \geq \log\left(\max_{t\colon \floor{ct}< n}\frac{\binom{n}{t}}{\binom{\floor{ct}}{t}}\right)\,.$$ 
This proves (\ref{LowerIn1}).
In order to prove (\ref{LowerIn2}), first note that by Lemma~\ref{simplelemma}, 
 $$\log\left( 1+ \frac{(c-1)^{c-1}}{c^c} \right) n \leq \frac{n}{c}\,.$$
Thus, for $c \geq \frac{n}{2}$, the righthand side of (\ref{LowerIn2}) is negative, and hence, the inequality is trivially true.

Assume now that $c < \frac{n}{2}$.
We will determine an integer value of $t$ such that $\binom{n}{t}/\binom{\floor{ct}}{t}$ becomes sufficiently large. First, we use Lemma~\ref{entro}: 
\begin{align*}
\frac{\binom{n}{t}}{\binom{\floor{ct}}{t}}&\geq \frac{2^{n H(t/n)}}{(n+1)\cdot2^{\floor{ct} H(t/\floor{ct})}}
=\frac{2^{nH(t/n)-\floor{c t} H(t/\floor{ct})}}{n+1}
\end{align*}
It is possible that $t=\floor{ct}$, but this is fine since $H(1)=0$. Using {\em (H4)}, we see that
\begin{align*}
\floor{c t} H\left(\frac{t}{\floor{ct}}\right)&\leq ct H\left(\frac{t}{ct}\right)
= ct H\left(\frac{1}{c}\right).
\end{align*}
Thus,
\begin{equation}
\log\left(\frac{\binom{n}{t}}{\binom{\floor{ct}}{t}}\right)\geq nH\left(\frac{t}{n} \right)- c t H\left(\frac{1}{c} \right)-\log(n+1)= M(n,t) - \log (n+1).
\label{logbn}
\end{equation}
Let $t'=\frac{n}{x}$. We know that $M(n,t)$ attains its maximum value when $t=t'$. Since $c>1$, it is clear that $x>c$ and hence $t'<\frac{n}{c}$. It follows that $\floor{ct'}<n$. However, $t'$ might not be an integer. In what follows, we will first argue that $\floor{c\ceil{t'}}<n$ and then that $M(n,\ceil{t'})$ is close to $M(n,t')$. The desired lower bound will then follow by setting $t=\ceil{t'}$. 

Using calculus, it can be verified that, for $c>1$, $x/c$ is increasing in $c$.
Hence,
\begin{align*}
\frac{x}{c} 
& = \left(\frac{c}{c-1}\right)^{c-1}+\frac{1}{c}\\
& \geq \lim_{c\rightarrow 1^+} \left( \left(\frac{c}{c-1}\right)^{c-1}+\frac{1}{c} \right), \text{ for } c>1\\
& = \lim_{c\rightarrow 1^+} \left(1+\frac{1}{c-1}\right)^{c-1}+ \lim_{c\rightarrow 1^+} \frac{1}{c}
 = \lim_{a\rightarrow 0^+} \left(1+\frac{1}{a}\right)^{a} +1
 = 2\,.
\end{align*}
Thus, $c \leq x/2$, and hence,
$$\floor{c\ceil{t'}} \leq c \left\lceil \frac{n}{x} \right\rceil < \frac{cn}{x}+c \leq \frac{n}{2}+c < n\,. $$

Note that $\frac{d}{dt}M(n,t)<0$ for $t>t'$, so $M(n,\ceil{t'})\geq M(n,t'+1)$. Combining this observation with (H2) and (H5), we get that
\begin{align*}
M(n,\ceil{t'})&\geq M(n, t'+1)
=nH\left(\frac{t'+1}{n}\right)-c(t'+1)H\left(\frac{1}{c}\right)\\
&=nH\left(\frac{1}{x}+\frac{1}{n}\right)-c\frac{n}{x}H\left(\frac{1}{c}\right)-cH\left(\frac{1}{c}\right)\\
&\geq\; nH\left(\frac{1}{x}+\frac{1}{n}\right)-c\frac{n}{x}H\left(\frac{1}{c}\right)-\log n-2, \text{ by {\em (H2)}}\\
&\geq\; nH\left(\frac{1}{x}\right)-c\frac{n}{x}H\left(\frac{1}{c}\right)-\log n-5, \text{ by {\em (H5)}}\\
&= M(n,t')-\log n-5.
\end{align*}
By choosing $t=\ceil{t'}$ in the $\max$, we conclude that

\begin{align*}
\log\left(\max_{t\colon \floor{ct}< n}\frac{\binom{n}{t}}{\binom{\floor{ct}}{t}}\right)&\geq M(n,\ceil{t'})-\log (n+1), \text{ by~(\ref{logbn})}\\
&\geq M(n,t')-\log(n+1)-\log n-5\\
&\geq n\log\left(1+\frac{(c-1)^{c-1}}{c^{c}}\right)-2\log (n+1)-5, \text{ by~(\ref{mcalc})}.
\end{align*}
\end{proof}

The following lemma is used for proving Theorem~\ref{smallamin}.

\begin{lemma}
\label{finallemma}
If $c$ is an integer-valued function of $n$ and $c>1$, it holds that
\begin{equation*}
\log\left(\max_{t\colon ct< n}\frac{\binom{n}{t}}{\binom{ct}{t}}\right)=\Omega\left(\frac{n}{c}\right).
\end{equation*}
\end{lemma}
\begin{proof}
Assume that $c$ is an integer-valued function of $n$, that $c>1$ and that $ct<n$. It follows that
\begin{align*}
\frac{\binom{n}{t}}{\binom{ct}{t}}&=\frac{n!(ct-t)!}{(n-t)!(ct)!}
\geq\frac{n (n-1)\cdots (n-t+1)}{(ct)(ct-1)\cdots (ct-t+1)}
\end{align*}
Let $t=\floor{\frac{n}{ec}}$. Then
\begin{align*}
\frac{\binom{n}{t}}{\binom{ct}{t}}&=\frac{n (n-1)\cdots (n-t+1)}{(ct)(ct-1)\cdots (ct-t+1)}
\geq \frac{n (n-1)\cdots (n-t+1)}{\frac{n}{e}(\frac{n}{e}-1)\cdots (\frac{n}{e}-t+1)}\\
&= \frac{n}{\frac{n}{e}}\frac{n-1}{\frac{n}{e}-1}\cdots\frac{n-t+1}{\frac{n}{e}-t+1}
\geq e^t.
\end{align*}
Since 
\begin{align*}
\log(e^t)=t\log(e)\geq \left(\frac{n}{ec}-1\right)\log e = \frac{n}{e \, \ln(2) \, c}-\log(e)=\Omega\left(\frac{n}{c}\right),
\end{align*} 
this proves the lemma by choosing $t=\floor{\frac{n}{ec}}$.
\end{proof}

\subsection{Approximating the Advice Complexity Bounds for \maxS}
\label{approxmax}

Lemma~\ref{calclowerM} of this section is used for Theorems~\ref{maxsgalg}--\ref{kpacklow}.
In proving Lemma~\ref{calclowerM}, the following lemma will be useful.

\begin{lemma}
\label{binomequal}
For all $n,c$, it holds that
\begin{equation*}
\label{mmbino1}
\max_{u \colon 0<u<n}\frac{\binom{n}{u}}{\binom{n-\ceil{u/c}}{n-u}}\leq n \left(\max_{t\colon \floor{ct}< n}\frac{\binom{n}{t}}{\binom{\floor{ct}}{t}}\right).
\end{equation*}
On the other hand, it also holds that
\begin{equation*}
\label{mmbino2}
\max_{u \colon 0<u<n}\frac{\binom{n}{u}}{\binom{n-\ceil{u/c}}{n-u}}\geq \frac{1}{n}\left(\max_{t\colon \floor{ct}< n}\frac{\binom{n}{t}}{\binom{\floor{ct}}{t}}\right).
\end{equation*}
\end{lemma}

\begin{proof}
Let $$f_{n,c}(t)=\frac{\binom{n}{t}}{\binom{\floor{ct}}{t}} \: \text{ and } \: g_{n,c}(u)=\frac{\binom{n}{u}}{\binom{n-\ceil{u/c}}{n-u}}\,.$$ In order to prove the upper bound, we show that $f_{n,c}(\floor{u/c}) \geq g_{n,c}(u)/n$, for any integer $u$, $0<u<n$. Note that $\floor{u/c}<u$, since $c>1$.
\begin{align*}
f_{n,c}(\floor{u / c})&=\frac{\binom{n}{\floor{u / c}}}{\binom{\floor{c\floor{u/c}}}{\floor{u / c}}}\\
&\geq\; \frac{\binom{n}{\floor{u / c}}}{\binom{u}{\floor{u / c}}}, \text{ by } \Bincr\\
&=\;\frac{\binom{n}{u}}{\binom{n-\floor{u /c}}{n-u}}, \text{ by } \Bfrac\\
&\geq\;\frac{u-\floor{u/c}}{n-\floor{u/c}} \, \frac{\binom{n}{u}}{\binom{n-\ceil{u /c}}{n-u}}, \text{ by } \Bminus\\
&\geq\;\frac{u-\floor{u/c}}{n-\floor{u/c}} \: g_{n,c}(u)\\
&\geq\frac{1}{n} \, g_{n,c}(u), \text{ since } u-\floor{u/c} \geq 1.
\end{align*}
By \Bminus, the second last inequality is actually an equality, unless $u/c$ is an integer.

 In order to prove the lower bound, we will show that $g_{n,c}(\ceil{ct}) \geq f_{n,c}(t)/n$, for any integer $t$ with $\floor{ct}< n$. Note that $t<\ceil{ct}$, since $c>1$.
\begin{align*}
g_{n,c}(\ceil{ct})&=\frac{\binom{n}{\ceil{ct}}}{\binom{n-\ceil{\ceil{ct}/c}}{n-\ceil{ct}}}\\
&\geq\; \frac{\binom{n}{\ceil{ct}}}{\binom{n-t}{n-\ceil{ct}}}, \text{ by \Bincr}\\
&=\frac{\binom{n}{n-\ceil{ct}}}{\binom{n-t}{n-\ceil{ct}}}\\
&=\;\frac{\binom{n}{n-t}}{\binom{\ceil{ct}}{t}}, \text{ by \Bfrac}\\
&=\frac{\binom{n}{n-t}}{\frac{\ceil{ct}}{\ceil{ct}-t}\binom{\floor{ct}}{t}}, \text{ by \Bminus}\\
&=\frac{\ceil{ct}-t}{\ceil{ct}} \, f_{n,c}(t)\\
&\geq\frac{1}{n} \, f_{n,c}(t), \text{ since } \ceil{ct}-t \geq 1 \text{ and } \ceil{ct} \leq n.
\end{align*}
\end{proof}

\begin{lemma}
\label{calclowerM}
Let $c>1$ and $n\geq 3$. It holds that
\begin{align*}
\log\left(\max_{u\colon 0<u< n} C(n,n-\left\lceil\frac{u}{c}\right\rceil,n-u)\right)&\geq \log\left(\max_{u \colon 0<u<n}\frac{\binom{n}{u}}{\binom{n-\ceil{\frac{u}{c}}}{n-u}}\right)\\
&\geq \log\left(1+\frac{(c-1)^{c-1}}{c^{c}}\right)n-3\log n-6.
\end{align*}
Furthermore, 
\begin{align*}
\log\left(\max_{u\colon 0<u< n} C(n,n-\left\lceil\frac{u}{c}\right\rceil,n-u)\right)&\leq \log\left(\max_{u \colon 0<u<n}\frac{\binom{n}{u}}{\binom{n-\ceil{\frac{u}{c}}}{n-u}} \, n\right)\\
&\leq \log\left(1+\frac{(c-1)^{c-1}}{c^{c}}\right)n+4\log(n+1)
\end{align*}
\end{lemma}

\begin{proof}
We prove the lower bound first.
\begin{align*}
&\log\left(\max_{u\colon 0<u< n} C(n,n-\left\lceil\frac{u}{c}\right\rceil,n-u)\right)\\
\geq & \log\left(\max_{u \colon 0<u<n}\frac{\binom{n}{u}}{\binom{n-\ceil{\frac{u}{c}}}{n-u}}\right), \text{ by Lemma ~\ref{erdos}}\\
\geq & \log\left(\max_{t\colon \floor{ct}< n}\frac{\binom{n}{t}}{\binom{\floor{ct}}{t}}\right)- \log n, \text{ by Lemma~\ref{binomequal}}\\
\geq & \log\left(1+\frac{(c-1)^{c-1}}{c^{c}}\right)n-2\log(n+1)-5 - \log n, \text{ by (\ref{LowerIn2})}\\
\geq & \log\left(1+\frac{(c-1)^{c-1}}{c^{c}}\right)n-3\log n-6, \text{ since } n \geq 3.
\end{align*}

We now prove the upper bound.
\begin{align*}
& \log\left(\max_{u\colon 0<u< n} C(n,n-\left\lceil\frac{u}{c}\right\rceil,n-u)\right)\\
& \leq \log\left(\max_{u \colon 0<u<n}\frac{\binom{n}{u}}{\binom{n-\ceil{\frac{u}{c}}}{n-u}}\left(1+\ln\binom{n-\ceil{u/c}}{n-u}\right)\right), \text{ by Lemma~\ref{erdos}}\\
& \leq \log\left(\max_{u \colon 0<u<n}\frac{\binom{n}{u}}{\binom{n-\ceil{\frac{u}{c}}}{n-u}} \, n\right)\\
& = \log\left(\max_{u \colon 0<u<n}\frac{\binom{n}{u}}{\binom{n-\ceil{\frac{u}{c}}}{n-u}}\right) + \log n\\
& \leq \log\left(\max_{t\colon \floor{ct}< n}\frac{\binom{n}{t}}{\binom{\floor{ct}}{t}} \, n\right) + \log n, \text{ by Lemma~\ref{binomequal}}\\
&\leq \log\left(1+\frac{(c-1)^{c-1}}{c^{c}}\right)n+4\log (n+1), \text{ by (\ref{UpperIn2})} 
\end{align*}
\end{proof}


\end{document}